\documentclass[12pt]{article}
\usepackage{amsmath}
\usepackage{graphicx,psfrag,epsf}
\usepackage{enumerate,color}
\usepackage{todonotes}
\usepackage{natbib}
\usepackage{url} % not crucial - just used below for the URL 

\newcommand{\blind}{1}

\usepackage{amsthm}
\usepackage{graphicx,hyperref}
\usepackage{amssymb}
\usepackage{amsmath}
\usepackage{multirow}
\usepackage{float}
\usepackage{mathtools}
\mathtoolsset{showonlyrefs=true}

\usepackage{subfig}
\usepackage{graphicx}
\usepackage{epsfig}
\usepackage{epstopdf}
\usepackage{amsthm}

\newtheorem{theorem}{Theorem}
\newtheorem{preposition}{Proposition}

\newtheorem{definition}{Definition}

\date{}
\begin{document}

\bibliographystyle{natbib}

\def\spacingset#1{\renewcommand{\baselinestretch}%
{#1}\small\normalsize} \spacingset{1}

%%%%%%%%%%%%%%%%%%%%%%%%%%%%%%%%%%%%%%%%%%%%%%%%%%%%%%%%%%%%%%%%%%%%%%%%%%%%%%

\if1\blind
{
  \title{\bf On initial direction, orientation and discreteness in the analysis of circular variables}
    \author{Gianluca Mastrantonio \hspace{.2cm}\\
    Department of Economics, University of Roma Tre\\
    and \\
    Giovanna Jona Lasinio \\
    Department of Statistical Sciences, University of Rome ``Sapienza''\\
    and\\
    Antonello Maruotti\\
    Department of Economic, Political Sciences and Modern Languages, LUMSA\\
    and \\ Gianfranco Calise\\
    Department of Earth Science, University of Rome ``Sapienza''
    }
  \maketitle
 \fi
}
\if0\blind
{
  \bigskip
  \bigskip
  \bigskip
  \begin{center}
    {\LARGE\bf Title}
\end{center}
  \medskip
} \fi

\bigskip
\begin{abstract}
In this paper, we propose a discrete circular distribution obtained by extending the wrapped Poisson distribution.  This new distribution, the Invariant Wrapped Poisson (IWP), enjoys numerous advantages:  simple tractable  density,  parameter-parsimony  and  interpretability,  good  circular dependence structure and easy random number generation thanks to known marginal/conditional distributions. Existing discrete circular distributions strongly depend on the initial direction and orientation, i.e. a change of the reference system on the circle may lead to misleading inferential results.
We investigate the invariance properties, i.e. invariance under change of initial direction and of the reference system orientation, for several continuous and discrete distributions. We prove that the introduced IWP distribution satisfies these two crucial properties. We estimate parameters in a Bayesian framework and provide all computational details to implement the algorithm. Inferential issues related to the invariance properties are discussed through numerical examples on artificial and real data. \end{abstract}

\noindent%
\vfill

\newpage

% spazio originale
\spacingset{1.45} % DON'T change the spacing!

%\spacingset{1} 

\section{Introduction}
Circular data (for a review see e.g. \citealp{lee2010}) arise naturally in many scientific fields where observations are recorded as directions or angles. Examples of circular data include wind and wave directions \citep{Bulla2012,Lagona2014,wang2014,Lagona2014,lagona2015,mastrantonio2015,mastrantonio2015b}, animal movements
\citep{Eckert2008,Langrock2012,langrock2014b,McLellan2015}, social science 
\citep{Gill2010}  and  auditory localization data  \citep{McMillan2013}. Standard statistical distributions cannot be used to analyze circular data because of the finite and on the unit circle $[0,2\pi)$ support and to dependence of descriptive and inferential analyses on the initial value and orientation on the unit circle. These features make circular data special, so that ad-hoc method and distributions have been developed in the literature.

Dating back to \cite{vonmises1918}, the attention over circular data has increased over time \citep{mardia72,fisher1996,Merdia1999,Jammalamadaka2001,pewsey2013}, leading to important probability distributions theory and inferential results. Nevertheless, not all the introduced circular distributions hold/satisfy crucial properties required to avoid misleading inference, as their parameters strongly depend on the chosen initial direction and sense of orientation (clockwise or anti-clockwise) on the unit circle. Indeed, in general, any circular distribution should have good fitting characteristics, a tractable form, be parsimonious in terms of parameters at play and invariant with respect to the chosen reference system.

All these conditions are fulfilled by the new distribution we propose in the present paper. Motivated by real-data examples, we introduce a new discrete circular distribution, the Invariant Wrapped Poisson (IWP), and investigate its properties along with a discussion on the interpretation of parameters. The literature on discrete circular data modelling has been often overlooked as discrete circular distributions have several drawbacks, as we will discuss later in this paper. However, the increased availability of discrete data on the circle has recently led to a series of researches introducing distributions able to account for this specific data feature \citep{Girija2014,Jayakumar2012}. We would contribute to extend this branch of literature and provide a reference distribution for the analysis of discrete circular data. The Poisson distribution is recognized as an important tool to analyze count data, and its correspondent on the circle, the Wrapped Poisson distribution, is often considered for count data analysis on the circle. Nevertheless, as the Poisson distribution is restrictive because it assumes a unit, variance-to-mean ratio % \textbf{Questa cosa non mi convince, perchè quando passi alle variabili circolari, i.e. the wrapped Poisson, non \`e vero che c'\`e un rapporto 1 a 1, il rapporto tra media e varianza circolare non \`e neanche costante, ma pi\`u che altro vie \`e assoluta dipendenza: cio\`e tu dimmi la media circolare e io ti dico qual \`e la varianza circulara. Non so come si pu\`o dire, functional dependence? dependence? }, 
the Wrapped Poisson \citep{Merdia1999} shows a strict functional relation between circular mean and variance, and strongly depends on the reference system orientation adopted. To avoid these issues, and recognizing the importance of relaying on a Poisson-based distribution, we introduce the IWP which is invariant with respect to the initial origin and the reference system {and, furthermore, allows for relaxing the strong imposed linkage between circular mean and variance. Indeed, by deriving IWP's trigonometric moments, it is possible to show that  the circular mean and variance are independent.}

  { 	
We  estimate the model under a Bayesian framework proposing two Markov chain Monte Carlo (MCMC) algorithms. An exact algorithm is considered and, to increase efficiency, an approximated algorithm is introduced based on IWP density approximation. The  approximation is obtained investigating the link between  
the wrapped Normal  and the IWP densities and following the  truncation strategy proposed by \cite{Jona2013} for the wrapped Normal distribution. 
Both algorithms use the well known latent variable approach for wrapped distribution  of  \cite{Coles1998} and  a specific prior   distribution for one  of the parameters. The approximated algorithm requires  a re-parametrization of the latent variable to ensure computational efficiency and avoid local optima. }

The invariance properties in circular distributions are formally introduced in Section 2, along with the notation used throughout the paper. Proper definitions of invariance under changes of initial direction and reference system orientation are provided along with necessary and sufficient conditions that circular distributions must hold to be invariant. Section 3 is devoted to the discussion of these properties in widely used (continuous and discrete) circular distribution. In the same section we shall show that all the known discrete distributions do not hold the two invariance properties. The IWP is introduced in Section 4. Its good properties are widely discussed and all computational details provided. Model inference in a Bayesian setting is provided in Section 5. The good properties of the IWP distribution will be studied on artificial and real data from different empirical contexts, and compared with  well-known models from the literature in Section 6. Section 7 summarized obtained inferential results and provides a discussion of possible applications and extensions of the IWP distribution.

\section{ {Invariance in circular distributions}} \label{sec:icsico}

Let $\{ \mathbb{S},\mathcal{A},P \}$ be a probability space, where the sample space $\mathbb{S}= \{(x,y): x^2+y^2=1  \}$ is the unit circle,  $\mathcal{A}$ is the $\sigma-$algebra on $\mathbb{S}$ and $P:\mathbb{S} \rightarrow [0,1] $ is the normalized Lebesgue measure on the measurable space $\{ \mathbb{S},\mathcal{A}\}$.  
Let $\mathbb{D}$ be a subset of $ \mathbb{R}$ such that its length is $2 \pi$, for example if   $\mathbb{D} =[a,b)$ we have $b-a = 2 \pi$. 

 {Let us c}onsider the 
measurable function $\Theta:\mathbb{S} \rightarrow \mathbb{D}$, with 
$
\Theta^{-1}(d)= (x,y)=(\cos d,\sin d ),\, d \in \mathbb{D},
$
 {and} let $\mathcal{D}=\sigma(\mathbb{D})$)  be the sigma algebra of $\mathbb{D}$ induced by $\Theta$ and   $A_{\Theta,D} \equiv \{(x,y):\Theta(x,y)\in D\}$ and 
%\begin{equation}\label{eq:p}
$
\mathbb{P}_{\Theta}  (D) = P(\Theta^{-1}(D)) = P\left( A_{\Theta,D} \right),\, \forall D \in \mathcal{D}
$.  Then the measurable space induced by $\Theta$ is $(\mathbb{D},\mathcal{D},\mathbb{P}_{\Theta})$  with %and it is  a probability space since  it satisfies the Kolmogorov axioms:
\begin{enumerate}
	\item $\mathbb{P}_{\Theta}  (D) =  P\left( A_{\Theta,D} \right)\geq 0,\, \forall D \in \mathcal{D}$;
	\item $\mathbb{P}_{\Theta}  (\mathbb{D}) =  P\left( A_{\Theta,\mathbb{D}} \right)= 1$;
	\item   for any countable sequence of disjoint  sets $\{D_j\}_{j=1}^{\infty}$  of $\mathcal{D}$,
	$
	\mathbb{P}_{\Theta}\left( \cup_{j=1}^{\infty}D_j  \right) = P\left(A_{\Theta,\cup_{j=1}^{\infty}D_j}  \right) = P \left(\cup_{j=1}^{\infty}  A_{\Theta,D_j}\right)= \sum_{j=1}^{\infty}P\left( A_{\Theta,D_j} \right) = \sum_{j=1}^{\infty} \mathbb{P}_{\Theta}(D_j ),
	$
\end{enumerate}
\noindent i.e.  {$(\mathbb{D},\mathcal{D},\mathbb{P}_{\Theta})$ is a probability space}
%The third property holds only if $\mathbb{D}$ is an interval of length $2 \pi$.
It follows that  $\Theta$ is a random variable and  $\mathbb{P}_{\Theta}$ is  its  \emph{probability distribution}.  $\Theta$ represents  an  angle over the unit circle and it is called a \emph{circular random variable}.  
 {Accordingly,} for all $d\in \mathbb{D}$, $\Theta^{-1}(d)=\Theta^{-1}(d \mod 2\pi)$ and, without loss of generality, we can represent any circular variable in $[0,2\pi)$. $\mathbb{D}$ can be either continuous or  discrete. In the latter case,  we assume that it is composed of  $l$ distinct points  equally spaced  between two extremes again denoted  $a$ and $b$ and we write  $\mathbb{D} \equiv \{ a+2\pi j/l   \}_{j=0}^{l-1}$.\\
 {If} $\mathbb{D}$ is a continuous domain,  $\Theta $ is a continuous  circular variable and $\mathbb{P}_{\Theta}$ is the Lebesgue measure;  {on the other hand, if} $\mathbb{D}$ is discrete, $\Theta$ is a discrete  circular variable  or a \emph{lattice} variable \citep[see][]{Merdia1999}, and $\mathbb{P}_{\Theta}$ is the counting measure. 
In both cases we indicate with  $f_{\Theta}= d\mathbb{P}_{\Theta}/d P_{\Theta}:\mathbb{D}\rightarrow \mathbb{R}^+$  the Radon-Nicodym derivative  of $\mathbb{P}_{\Theta}$ (rnd), i.e.  
$
\mathbb{P}_{\Theta}(D)= \int_{D}f_{\Theta}d{P}_{\Theta}.
$
Hence, we denote by $f_{\Theta}$ the probability density function (pdf) of $\Theta$ including both situations, continuous and discrete.% whenever results apply to both.% Only if necessary are going to specify iwhen it is continuous  while if it is a lattice variable, $f_{\Theta}$ describes  the probability mass function (pmf). The function $f_{\Theta}$ is generally specified as depending on  a vector of parameters $\boldsymbol{\psi}\in \boldsymbol{\Psi}$ and we write   $f_{\Theta}(\theta| \boldsymbol{\psi})$.

In the representation of circular variables a key role is played by the the initial direction (the angle 0) and the sense of orientation (clockwise or anti-clockwise) of the domain.  {Both} are uniquely determined  by the  choice of the orthogonal reference system  on  the plane.  {Any s}tatistical tool for circular variables should be invariant  {with respect to} different choices of th{e reference orientation} system to avoid conflicting or  {misleading} conclusions.  {Accordingly}, $f_{\Theta}(\cdot| \boldsymbol{\psi})$ must be invariant for changes in the orientation of the system's axis and initial direction. 
%Different statisticians that try to describe  a phenomena can take different choices   and it is important that they can reach the same conclusions using the same statistical tools, e.g.,  the descriptive statistics or a model based approaches, regardless of the choice of the origin and orientation of the circle.  
%  Definitions  \ref {def:def1} and \ref{def:def2} and  Theorem \ref{theo:primo} below,  clarify the above ideas of invariance. 
 {The following Definitions formalize the invariance properties with respect to the reference system on the plane.}
%
%To describe a  
%
%
%
%
%
%The sense of orientation of the unit circle and since this is just an arbitrary decision then, given a set of observed circular variables, the inference should not depends on the  particular orthogonal system chosen. A valid probability model for circular variables should allow for the same description of a phenomenon regardless of the choice of the origin and orientation of the circle.. 
%
%
%
%
%
%

\begin{figure}[t!]
	\centering
	\emph{Wrapped skew normal}\\
	{\subfloat[]{\includegraphics[trim=80 55 55 55,clip,scale=0.3]{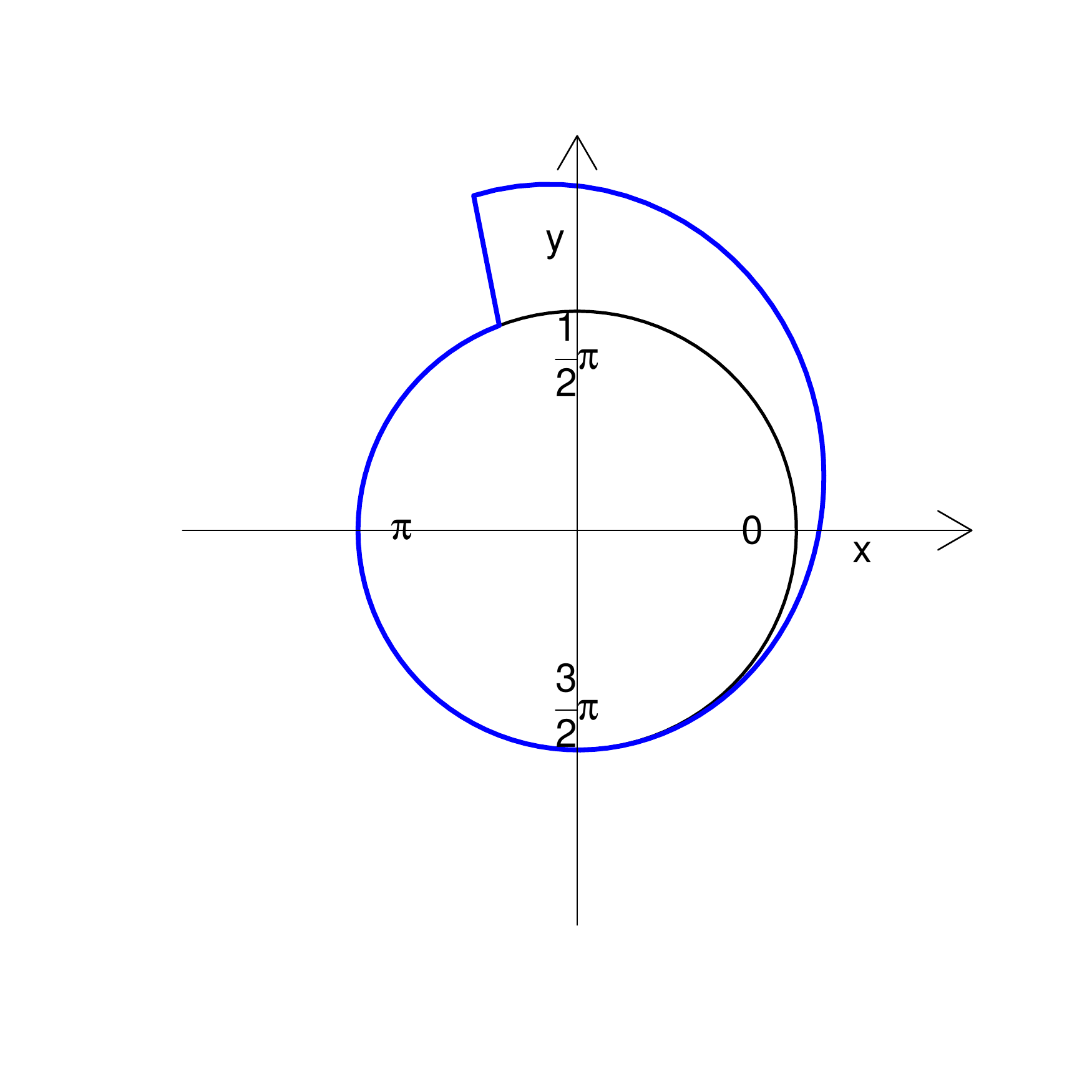}}}
	{\subfloat[]{\includegraphics[trim=80 55 55 55,clip,scale=0.3]{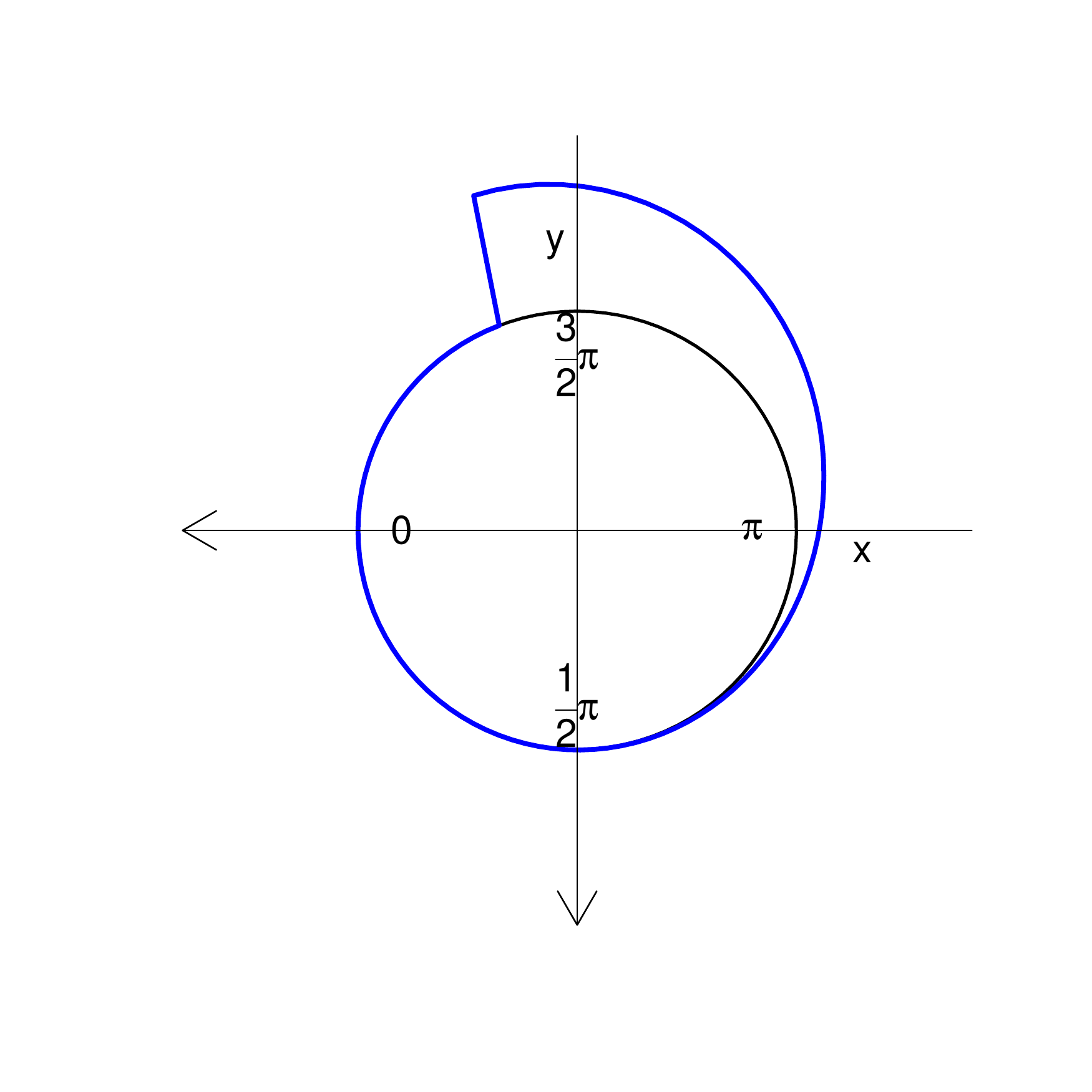}}}
	{\subfloat[]{\includegraphics[trim=80 55 55 55,clip,scale=0.3]{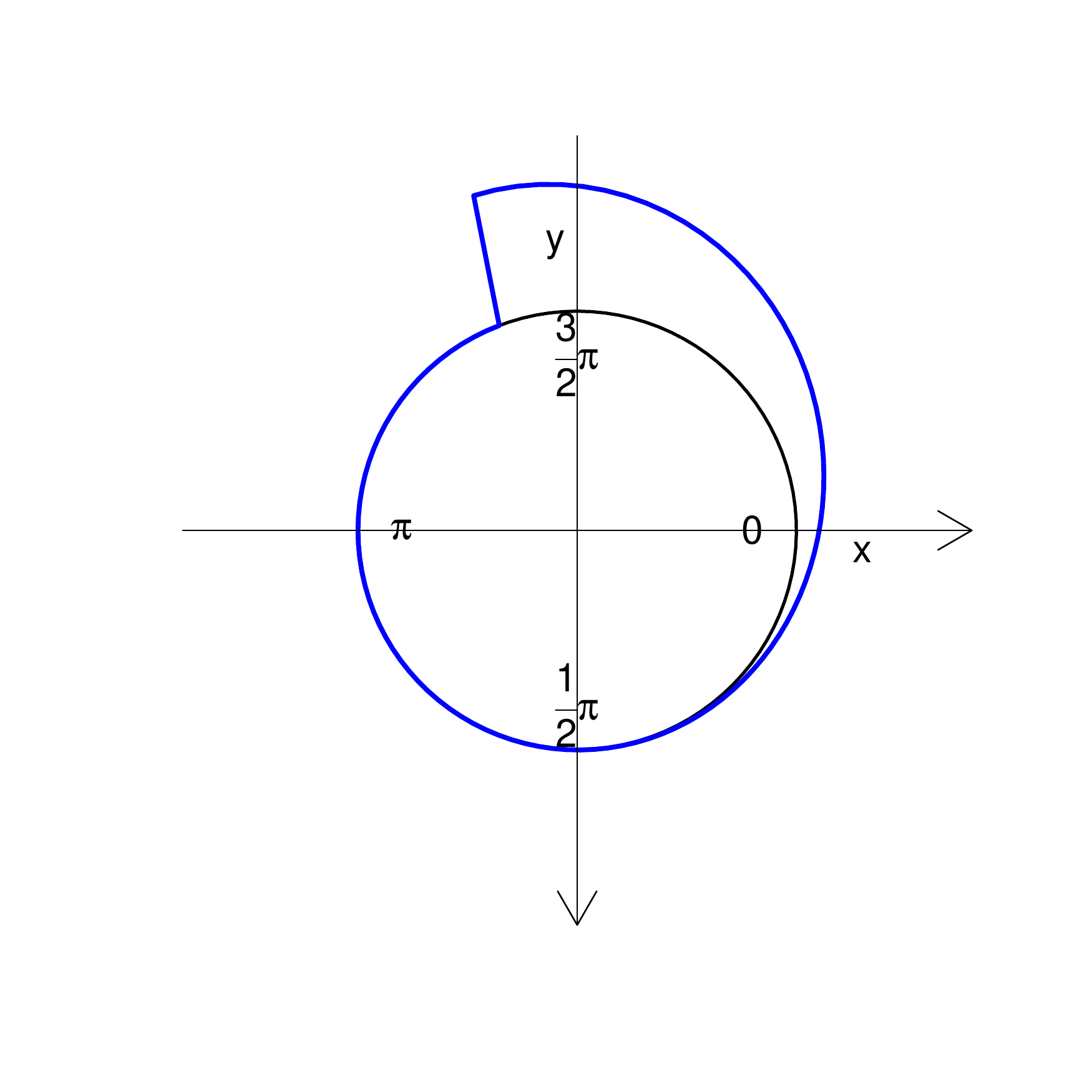}}}\\
	\emph{Wrapped exponential}\\
	{\subfloat[]{\includegraphics[trim=80 55 55 55,clip,scale=0.3]{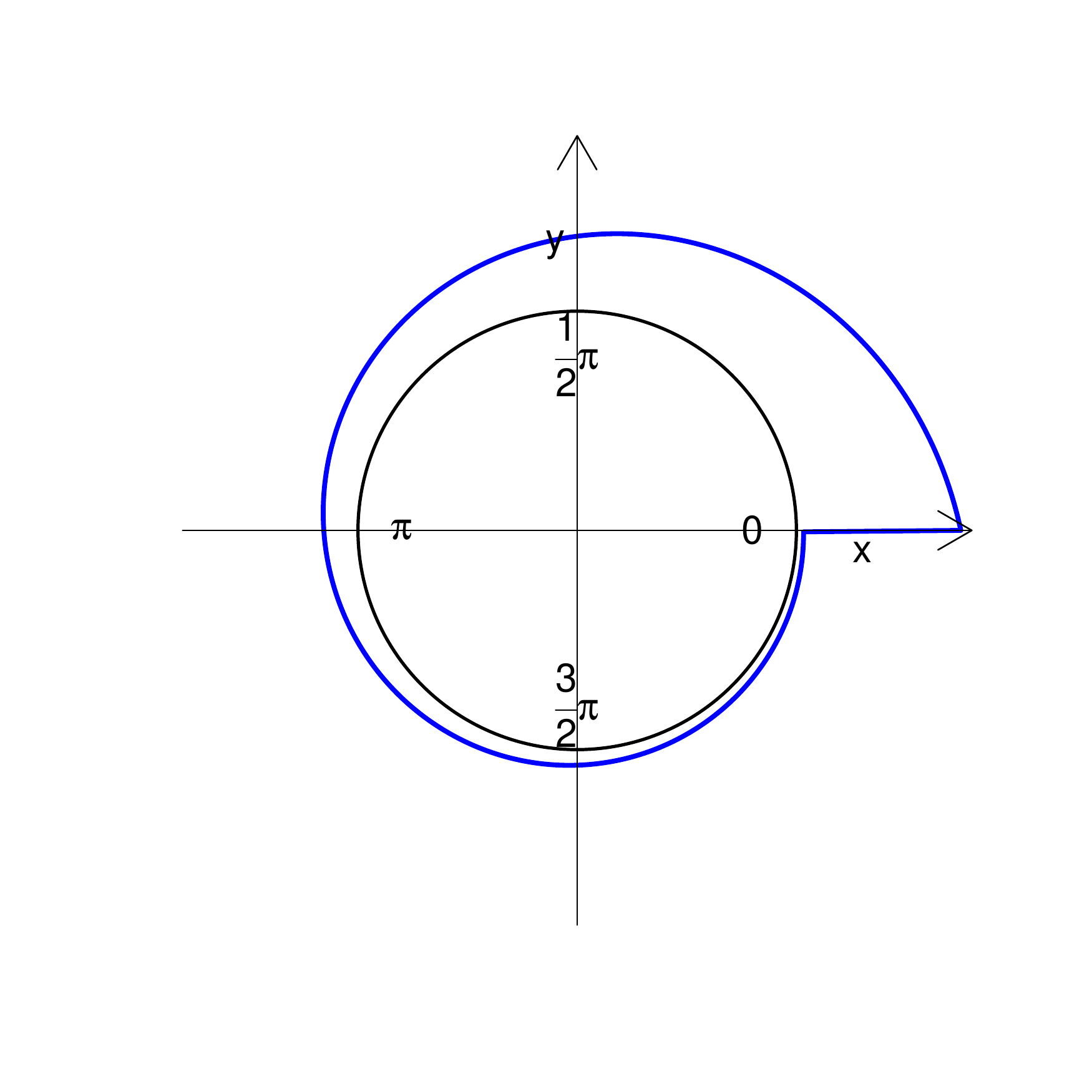}}}
	{\subfloat[]{\includegraphics[trim=80 55 55 55,clip,scale=0.3]{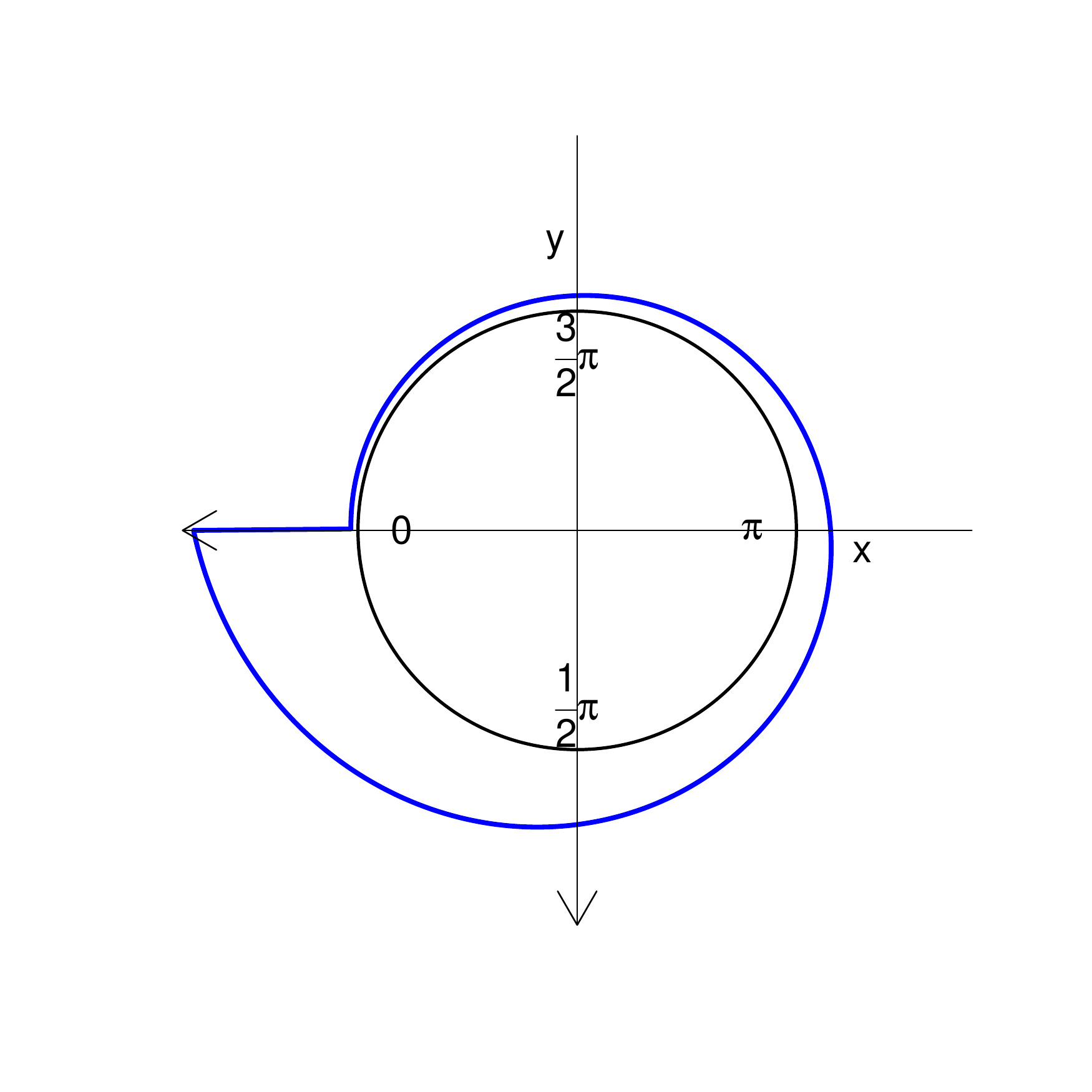}}}
	{\subfloat[]{\includegraphics[trim=80 55 55 55,clip,scale=0.3]{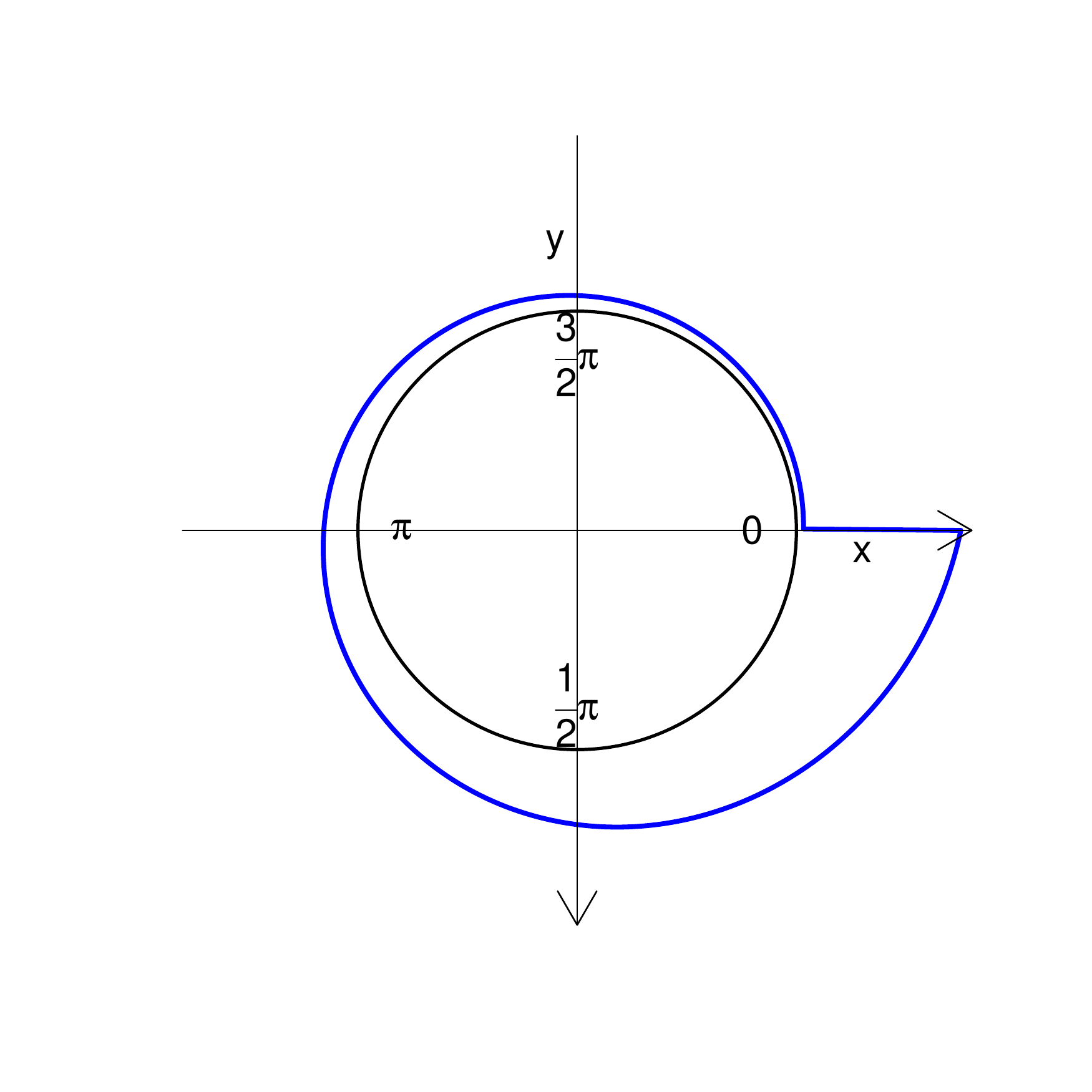}}}
	\caption{	Probability density functions of a wrapped skew normal (a-c) and a wrapped exponentil (d-f) under different initial directions and orientations. The arrows indicate the axis orientation } \label{fig:dens}
\end{figure}

\begin{definition}\label{def:def1}
	We say that $f_{\Theta}$ is invariant under change of initial direction  (ICID) if $\forall \xi \in \mathbb{D}$, $\forall \theta \in \mathbb{D}$ and $\forall \boldsymbol{\psi} \in \boldsymbol{\Psi}$ there exists ${\boldsymbol{\psi}^*}\in \boldsymbol{\Psi}$  such that
	%\begin{align} \label{eq:def1}
	$	f_{\Theta}(\theta| \boldsymbol{\psi} ) = f_{\Theta}(\theta- \xi	| {\boldsymbol{\psi}^*} ).
	$%\end{align}	
	%	that can be written equivalently as 
	%	\begin{align} \label{eq:def1}
	%	f_{\theta}(\Theta=\theta-_c \xi_2| \boldsymbol{\psi} ) = f_{\theta}(\Theta=\theta+_c\xi_1-_c\xi_2 	| \boldsymbol{\psi}^{*} ) ,\, \forall \xi_2 \in \mathbb{S}.
	%	\end{align}	
\end{definition}	
%Intuitively, Definition \ref{def:def1} means that we can move the distribution over the domain without changing its shape. For better understanding, consider the Figure \ref{fig:dens} (a) and suppose $\xi=\pi$ and then compare it with  Figure \ref{fig:dens} (b).

\begin{definition}\label{def:def2}
	We say that $f_{\theta}$ is invariant under changes of the reference system orientation  (ICO) if $\forall \xi \in \mathbb{D}$, $\forall \theta \in \mathbb{D}$ and $\forall \boldsymbol{\psi} \in \boldsymbol{\Psi}$ there exists $ {\boldsymbol{\psi}^*}\in \boldsymbol{\Psi}$  such that :
	%\begin{align}\label{eq:def2}
	$	f_{\Theta}(\xi-\theta| \boldsymbol{\psi} ) = f_{\Theta}(\xi+\theta| {\boldsymbol{\psi}^*}).
	$%\end{align}	
\end{definition}
%Definition \ref{def:def2} means that for each $\theta$ and $\boldsymbol{\psi}$ exist a vector of parameters that gives a specular distribution. For example, let $\xi=0$, it is easy to see that for each $\theta$ the values of the density of $(\xi+\theta)    \text{ mod }2 \pi$ in Figure  \ref{fig:dens} (a) is the same of $(\xi-\theta)    \text{ mod }2 \pi$ in Figure  \ref{fig:dens} (c). 
 {A simple example clarifies the two definition above. I}n Figure \ref{fig:dens} (a), a  wrapped skew normal density \citep{Pewsey2000},  {is plotted on} a sample of circular measurements. Most of the probability mass  is in the first quadrant of the current reference system.  The origin (initial direction) is set to zero radiant and the  orientation is anti-clockwise. Figure \ref{fig:dens} (b) is obtained by adding $\pi$ to Figure \ref{fig:dens} (a), i.e. changing the system origin, and Figure \ref{fig:dens} (c) is built as $-1\times$Figure \ref{fig:dens} (a) i.e. changing the system orientation. The densities in Figures \ref{fig:dens} (b) and (c) are still skew normal.  {The distribution in} \ref{fig:dens} is ICID and ICO.  {We would remark that a necessary, but not sufficient, condition to be ICID and ICO is that distributions under} different reference systems must have the same circular variance \citep{Merdia1999}, as for example shown in the three densities in Figures \ref{fig:dens} (a), (b) and  (c).  {To further corroborate this result, let us provide another example.} In figure \ref{fig:dens} (d) a wrapped exponential distribution \citep{Rao2004} is shown. %\cite{Rao2004} show that there is a one-to-one relation between  the parameter of the wrapped exponential and the circular variance. 
 {Let us transform circular measurements as before, in the wrapped skew normal examples. Although transformed densities in Figures \ref{fig:dens} (e) and (f) have the same variance as Figures \ref{fig:dens} (d), both look different from \ref{fig:dens} (d) at first glance. Indeed, the wrapped exponential is not ICID and ICO, as we will discuss in depth in the next section.}% the only densities that %have the same circular variances of the original wrapped exponential, Figure \ref{fig:dens} (d), under the new orthogonal %systems,  are the ones depicted in Figures \ref{fig:dens} (e) and (f). Clearly the wrapped exponential it is not ICID and ICO. \\

 {It is clear that a graphical inspection could reveal some indications, but we need to formally prove whether a circular distribution is ICID and ICO or not.} Theorem \ref{theo:primo} helps us to define distributions that are both ICID and ICO.
%evo modificato l'enunciato del teorema.}}
\begin{theorem} \label{theo:primo}
	Let $f_{\theta}(\cdot | \boldsymbol{\psi})$ be the pdf of the circular variable $\Theta\in \mathbb{D}$  with $\boldsymbol{\psi}\in\Psi$  and let $\Theta^*=g(\theta)=\delta(\theta+\xi)$ with scalar $\delta = \{-1,1\}$ , $\xi\in \mathbb{D}$ and pdf  $f_{\theta^*}(\cdot | \boldsymbol{\psi}^*)$ and $\boldsymbol{\psi}^*\in\Psi$. $f_{\theta}$ is ICID and ICO iff  %the density  of the  random variable $\Theta^*=g(\Theta) =\delta(\Theta+ \xi)$, with $\delta \in \{-1,1\}$ and $\xi \in \mathbb{D}$,  belong to the same parametric family as  $\Theta$, i.e.
	\begin{equation} \label{eq:theo1}
		f_{\Theta^*}(\theta^*| \boldsymbol{\psi}^*) = f_{\Theta}(\theta^* | \boldsymbol{\psi}^*).
	\end{equation}
i.e. $f_{\Theta}$ and $f_{\Theta^*}$ belong to the same parametric family.	 
\end{theorem}
%\textbf{Possiamo probabilmente fare un'appendice con tutte le proofs}
\begin{proof} 
	Let us focus on $f_{\Theta^*}$.  Notice that $g(\cdot)$ is a  linear mapping  such that  $\mathbb{D}\rightarrow \mathbb{D}$ and by  rule of variable transformation we have  that
	\begin{equation} \label{eq:2}
	f_{\Theta^*}(\theta^*| \boldsymbol{\psi}^*)=f_{\Theta}(\delta\theta^*-\xi| \boldsymbol{\psi}),
	\end{equation}
	where the vector of parameters $\boldsymbol{\psi}^*$ is a function of $\boldsymbol{\psi}$ and $(\delta,\xi)$.\\	
Now we prove that if 
	$f_{\Theta}$ and $f_{\Theta^*}$ belong to the same parametric family this  implies that $f_{\Theta}$ is ICID and ICO. %and then   we show that if $f_{\Theta}$ is ICID and ICO, the distributions of  $\Theta^*$ and $\Theta$ belong to the same parametric family.

	\begin{enumerate}
		\item {\bf If $f_{\Theta^*}(\theta^*| \boldsymbol{\psi}^*) = f_{\Theta}(\theta^* | \boldsymbol{\psi}^*)$ then \eqref{eq:theo1} implies
		\begin{equation} \label{eq:22}
		f_{\Theta}(\theta^* | \boldsymbol{\psi}^*)=f_{\Theta}(\delta\theta^*-\xi| \boldsymbol{\psi}),
		\end{equation} }
		Equation \eqref{eq:2} is true for all $\theta^* \in \mathbb{D}$, $\xi \in \mathbb{D}$, $\delta \in \{-1,1\}$ and moreover must be that $\boldsymbol{\psi}^* \in \boldsymbol{\Psi}$ since the left and right sides of \eqref{eq:22} belong to the same parametric family.
%		Let $\bar{\theta} = \theta^* \mbox{ mod } 2 \pi$. Since $\delta\theta^*-\xi = (\delta\bar{\theta}-\xi) \text{ mod }2 \pi$, using the equations \eqref{eq:ff}, \eqref{eq:theo1} and \eqref{eq:2} we can write
%		\begin{equation} \label{eq:ICID}
%		f_{\Theta}(\bar{\theta} | \boldsymbol{\psi}^*)  =  f_{\Theta}((\delta\bar{\theta}-\xi) \text{ mod }2 \pi| \boldsymbol{\psi})
%		\end{equation}
		To prove that $f_{\Theta}$ is ICID it is sufficient to set $\delta=1$ in \eqref{eq:22} and we  obtain  Definition \ref{def:def1}.  
		 For simplicity let in \eqref{eq:22}  $\xi=0$ and $\delta =-1$. Then we  immediately obtain   Definition \ref{def:def2}:
		$%\begin{equation} \label{eq:222}
		f_{\Theta}(\theta^* | \boldsymbol{\psi}^*)=f_{\Theta}(-\theta^*| \boldsymbol{\psi}),
		$%\end{equation} 
		As $f_{\Theta}$ is ICID what holds in $\xi=0$ is going to hold for all $\xi\in \mathbb{D}$ then:
		
		% but only for a specific value of $\xi$, in this case $\xi=0$. Since $f_{\Theta}$ is ICID for all $\xi^* \in \mathbb{D}$ exist $\boldsymbol{\psi}^{**} \in \boldsymbol{\Psi}$ and $\boldsymbol{\psi}^{***}\in \boldsymbol{\Psi}$  with the following property:
		\begin{equation}  \label{eq:ICO}
		f_{\Theta}(\theta^* | \boldsymbol{\psi}^*)=f_{\Theta}(\xi^*+\theta^* | \boldsymbol{\psi}^{**})=f_{\Theta}(\xi^*-\theta^*| \boldsymbol{\psi}^{***})=f_{\Theta}(-\theta^*| \boldsymbol{\psi}).
		\end{equation}
		The central part of \eqref{eq:ICO} is  Definition \ref{def:def2} and  proves  that $f_{\Theta}$ is also ICO.
%		Since $f_{\Theta}$ is ICID we can write  
%		\begin{equation} 
%		f_{\Theta}(\bar{\theta} | \boldsymbol{\psi}^*)= f_{\Theta}((\bar{\theta}-(M-\xi))  \text{ mod }2 \pi | \boldsymbol{\psi}^{**})= f_{\Theta}((\delta\bar{\theta}-\xi) \text{ mod }2 \pi| \boldsymbol{\psi})
%		\end{equation}
%		We set $\delta=-1$ and since $(-\bar{\theta}-\xi) \text{ mod }2 \pi = (\bar{\theta}+(M-\xi))\text{ mod }2 \pi$  follows that 
%		\begin{equation} 
%		f_{\Theta}((\bar{\theta}-(M-\xi))  \text{ mod }2 \pi | \boldsymbol{\psi}^{**})= f_{\Theta}((\bar{\theta}+(M-\xi))\text{ mod }2 \pi| \boldsymbol{\psi}).
%		\end{equation}
%		We have that $(M-\xi)  \in \mathbb{D}$ and then  $f_{\Theta}$ is ICO. 
		\item\textbf{ If $f_{\Theta}$ is ICID and ICO then
			%\begin{align}
				$f_{\Theta^*}(\theta^*| \boldsymbol{\psi}^*) = f_{\Theta}(\theta^*| \boldsymbol{\psi}^*)$.
			%\end{align}
			}	\\
		Since $f_{\Theta}$ is  	ICID and ICO we can write $f_{\Theta}(\delta\theta^*-\xi| \boldsymbol{\psi}) = f_{\Theta}(\theta^*| \boldsymbol{\psi}^{**})$ with   $ \boldsymbol{\psi}^{**} \in \boldsymbol{\Psi}$. Using this equivalence with equation  \eqref{eq:2} we can write
		$%\begin{equation} \label{th:1}
		f_{\Theta^*}(\theta^*| \boldsymbol{\psi}^{*}) = f_{\Theta}(\theta^*| \boldsymbol{\psi}^{**}).
		$%\end{equation}
		Since  $f_{\Theta^*}$ and $f_{\Theta}$ assume the same value at the same point for all $\theta \in \mathbb{D}$   they must be the same function with  $\boldsymbol{\psi}^*\equiv \boldsymbol{\psi}^{**}$.

%		
%		
%		\begin{equation} 
%		f_{\Theta^*}(\theta^*| \boldsymbol{\psi}^*)=f_{\Theta}(\delta\theta^*-\xi| \boldsymbol{\psi}).
%		\end{equation}	
%		We can write  \eqref{eq:2} as
%		%\begin{eqnarray} 
%		$f_{\theta^*}(\Theta^*=\theta^*| \boldsymbol{\psi}^*)=f_{\theta}(\Theta=(\delta\theta^*-\xi) \text{ mod }2 \pi| \boldsymbol{\psi})$
%		%\end{eqnarray}
%		and since $f_{\theta}(\cdot | \cdot)$ is ICS and ICO we have that exist a set of parameters $\boldsymbol{\psi}^{**}$ since that  $f_{\theta}(\Theta=(\delta\theta^*-\xi) \text{ mod }2 \pi| \boldsymbol{\psi}) = f_{\theta}(\Theta=\theta^*| \boldsymbol{\psi}^{**}) $ and then 
%		\begin{eqnarray} 
%		f_{\theta^*}(\Theta^*=\theta^*| \boldsymbol{\psi}^*)=f_{\theta}(\Theta=\theta^*| \boldsymbol{\psi}^{**}),  \, \forall \theta^* \in \mathbb{S}
%		\end{eqnarray}
%		and then must be $\boldsymbol{\psi}^*=\boldsymbol{\psi}^{**}$ and $f_{\theta^*}(\Theta^*=\theta^*| \boldsymbol{\psi}^*) = f_{\theta}(\Theta=\theta^*| \boldsymbol{\psi}^*)$. 
		\end{enumerate}	
		
	\end{proof}

%We can obtain a pdf that verifies theorem \ref{theo:primo} starting from one that it does not. Our proposal is to include the quantities $\delta$ and $\xi$ into the parameters set characterizing the pdf of the circular variable, this simple transformation induces a new pdf that is ICID and ICO. Proposition \ref{prep:1} formalizes this procedure.

 {It is always possible to transform any circular pdf so to obtain its  ICID and ICO version.}

% If $f_{\Theta}$ doesn't have the two desired properties, we can apply a simple modification to the definition of the function and obtain an invariant distribution. The proof is given in the following proposition:  

%\textbf{Forse cos\`i si capisce meglio}
\begin{preposition} \label{prep:1}
	If  $f_{\Theta}(\cdot|\boldsymbol{\psi})$ is not ICID and ICO,  we can define $f_{\Theta^*}(\cdot|\boldsymbol{\psi}^*)$, where $\Theta^*= \delta(\Theta+\xi)$ and $\boldsymbol{\psi}^*=\{\boldsymbol{\psi}, \delta,\xi\}$ such that $f_{\Theta^*}(\cdot|\boldsymbol{\psi}^*)$ is ICID and ICO.%that verifies the two desired properties. In other words we include  $\delta$ and $\xi$ in the parameters set characterizing the circular distribution.% parameters of $f_{\Theta^*}$. \textbf{Mi sa che questa preposizione, soprattutto la dimostrazione, si capisca veramente poco. GIANLUCA NON CAPISCO MANCO COSA VUOI DIRE IN INGLESE ...  }
\end{preposition}
\begin{proof}
	% {Non mi \'e chiara, ma probabilmente \'e colpa mia. In particolare la parte con i tre $^***$.}
	Let $g(\Theta)= \delta(\Theta+\xi)$,
	we have  that $f_{\Theta^*}( \theta^*|\boldsymbol{\psi}^*)=f_{\Theta}(\delta\theta^*-\xi |\boldsymbol{\psi})$. \\
	 {Let us define $\Theta^{**} = \delta^{*}(\Theta^*+\xi^*)$, following  Theorem \ref{theo:primo}, if $f_{\Theta^{**}}$ belongs to the same parametric family as $f_{\Theta^{*}}$, then $f_{\Theta^*}$ is ICID and ICO. As} $f_{\Theta^{**}}( \theta^{**}|\boldsymbol{\psi}^{**})=f_{\Theta^*}(\delta^*\theta^{**}-\xi^{*} |\boldsymbol{\psi}^*)$ we have the following identity: 
	$	f_{\Theta^{**}}(\theta^{**}|\boldsymbol{\psi}^{**})=f_{\Theta^*}(\delta^*\theta^{**}-\xi^* |\boldsymbol{\psi}^*)=f_{\Theta}(\delta(\delta^*\theta^{**}-\xi^*)-\xi |\boldsymbol{\psi}).
	$ 		 
	Now   let $\delta^{**} = \delta \delta^*$ and $\xi^{**}= (\delta \xi^*+\xi  ) $   we can write 
	$%\begin{equation} \label{eq:ss}
	f_{\Theta^{**}}(\theta^{**}|\boldsymbol{\psi}^{**})=f_{\Theta^*}(\delta^*\theta^{**}-\xi^* |\boldsymbol{\psi}^*)=f_{\Theta}(\delta^{**}\theta^{**}-\xi^{**} |\boldsymbol{\psi}),
	$%\end{equation} 
	 {i.e.} $f_{\Theta^{**}}$ is of the same functional form as $f_{\Theta^*}$ and this proves theorem \ref{theo:primo}.
	%		   $\Theta^{**}=\delta^{**}(\Theta^*+\xi^{**})$  then  $f_{\Theta^*}(\theta^{**}|\boldsymbol{\psi}^{***}) = f_{\Theta}(\delta^{**}\theta^{**}-\xi^{**} |\boldsymbol{\psi})$ and this implies that   $f_{\Theta^{**}}(\theta^{**}|\boldsymbol{\psi}^{**})=f_{\Theta^*}(\theta^{**}|\boldsymbol{\psi}^{***}) $. Notice that given any  point $\theta^{**}$,  $f_{\Theta^{**}}\equiv f_{\Theta^{*}}$ coincide then they must be the same function with the same parameters set, (i.e. $f_{\Theta^{**}}(\theta^{**}|\boldsymbol{\psi}^{**})=f_{\Theta^*}(\theta^{**}|\boldsymbol{\psi}^{**})$), that is  the requirement of Theorem \ref{theo:primo}. 	
\end{proof}
 {Distributions obtained from Proposition \ref{prep:1} are said to be \emph{invariant}. In other words, by applying the previous proposition, we get the \emph{invariant version} of the density of $\Theta$.}
	
\section{ {Circular distributions:} { Examples} }

 {In this section  conventional circular distributions are investigated, focusing on the invariance properties previously introduced. To simplify the notation, we avoid the use of the modulus $2\pi$ operator, and define $\Theta^*= \delta(\Theta+\xi) \in \delta(\mathbb{D}+\xi)$.}
%In this Section we show how to apply Theorem   \ref{theo:primo} and we demonstrate that not all the distributions proposed in the literature are ICO and ICID.  The application of the Theorem (and the notation) is simplified if we consider $\Theta \in \mathbb{D}$ and $\Theta^*= \delta(\Theta+\xi) \in \delta(\mathbb{D}+\xi)$, i.e. we do not use the modulo $2 \pi$ operator.  

\subsection{Continuous distributions}

\subsubsection*{Wrapped Normal}
A circular variable $\Theta$ is said to have wrapped normal distribution with parameters $\mu$ and $\sigma^2$ ($\Theta \sim WN(\mu, \sigma^2)$)  if 
\begin{equation} \label{eq:wn}
f_{\Theta}(\theta| \mu, \sigma^2)=\sum_{k = -\infty}^{\infty} \left(2 \pi \sigma^2\right)^{-\frac{1}{2}}e^{\frac{(\theta+2 \pi k -\mu)^2}{2 \sigma^2}   }. 
\end{equation}
The random variable $\Theta^*$ has density
\begin{equation} \label{eq:wn2}
f_{\Theta^*}(\theta^*| \mu, \sigma^2, \delta, \xi)=\sum_{k = -\infty}^{\infty} \left(2 \pi \sigma^2\right)^{-\frac{1}{2}}e^{\frac{(\delta\theta^*-\xi+2 \pi k -\mu)^2}{2 \sigma^2}   }=\sum_{k = -\infty}^{\infty} \left(2 \pi \sigma^2\right)^{-\frac{1}{2}}e^{\frac{(\theta^*+2 \pi k -\delta(\xi+\mu))^2}{2 \sigma^2}   }.
\end{equation}
To prove that \eqref{eq:wn} is ICID and ICO, following Theorem \ref{theo:primo}, we need to prove  that  \eqref{eq:wn2} is still a wrapped normal density, i.e. there must exist two parameters $\mu^{*}$ and $\sigma^{2*}$, apart from the trivial case $\mu^{*}=\mu$ and $\sigma^{2*}=\sigma^2$, that are such that:  
\begin{equation} \label{eq:wn3}
\sum_{k = -\infty}^{\infty} \left(2 \pi \sigma^2\right)^{-\frac{1}{2}}e^{\frac{(\delta\theta^*-\xi+2 \pi k -\mu)^2}{2 \sigma^2}   } = \sum_{k = -\infty}^{\infty} \left(2 \pi \sigma^{2*}\right)^{-\frac{1}{2}}e^{\frac{(\theta^*+2 \pi k -\mu^*)^2}{2 \sigma^{2*}}   }.
\end{equation}
Clearly \eqref{eq:wn3} is true for $\mu^{*} = \delta(\xi+\mu)$ and $\sigma^{2*}=\sigma^{2}$, i.e. the wrapped normal is ICID and ICO. \\ 

\subsubsection*{Von Mises}
A circular variable $\Theta$ is said to follow a von Mises distribution with parameters $\mu$ and $\kappa$ ($\Theta \sim vM(\mu, \kappa)$)  if 
$%\begin{equation} \label{eqvm}
f_{\Theta}(\theta| \mu, \kappa) = \frac{e^{\kappa \cos (\theta-\mu)}}{2 \pi I_0(\kappa)}
$%\end{equation}
where $I_0(\kappa)$ is the modified Bessel function of order 0. As we did above for the wrapped normal,  we find the density of $\Theta^*$, that is
$%\begin{equation} \label{eqvm2}
f_{\Theta^*}(\theta^*| \mu, \kappa,\delta,\xi) = \frac{e^{\kappa \cos (\delta \theta^*-\xi -\mu)}}{2 \pi I_0(\kappa)} ,
$%\end{equation}
and since $\mu^*=\delta(\xi+\mu)$ and $\kappa^*=\kappa$ are such that: 
$%\begin{equation}
\frac{e^{\kappa \cos (\delta \theta^*-\xi -\mu)}}{2 \pi I_0(\kappa)}  = \frac{e^{\kappa^* \cos (\theta^*-\mu^*)}}{2 \pi I_0(\kappa^*)}, 
$ %\end{equation}
then the von Mises is ICID and ICO
\subsubsection*{Wrapped skew normal }
Suppose that $\Theta$ is distributed as a wrapped skew normal with parameters $\mu$, $\sigma^2$ and $\lambda$ ($\Theta \sim WSN(\mu, \sigma^2,\alpha)$), then
\begin{equation} \label{eq:wsn}
f_{\Theta}(\theta|\mu, \sigma^2,\alpha) = \sum_{k = -\infty}^{\infty} \frac{1}{\pi \sigma}e^{\frac{(\theta+2 \pi k -\mu)^2}{2 \sigma^2}   } \int_{-\infty}^{\alpha \left( \frac{\theta+2 \pi k -\mu}{\sigma}  \right)} e^{-\frac{t^2}{2}}dt.
\end{equation}
Since the distribution of $\Theta^*$ is 
$%\begin{equation} \label{eq:wsn2}
f_{\Theta}(\theta|\mu, \sigma^2,\alpha) = \sum_{k = -\infty}^{\infty} \frac{1}{\pi \sigma}e^{\frac{(\delta\theta^*-\xi+2 \pi k -\mu)^2}{2 \sigma^2}   } \int_{-\infty}^{\alpha \left( \frac{\delta\theta^*-\xi+2 \pi k -\mu}{\sigma}  \right)} e^{-\frac{t^2}{2}}dt,
$%\end{equation}
we can easily see that $\Theta^*\sim WSK(\delta(\xi+\mu), \sigma^2,\alpha)$ and then  \eqref{eq:wsn} is ICO and ICID. \\
\subsubsection*{Wrapped Exponential}
Not all distributions for  continuous circular variables have the two  required properties. For example the wrapped exponential of \cite{Rao2004} with density
\begin{equation} \label{eq:we}
f_{\Theta}(\theta| \lambda) = \frac{\lambda e^{-\lambda \theta}}{1-e^{-2 \pi \lambda}}, \, \lambda>0
\end{equation}
it is such that the  density of $\Theta^*$ is
\begin{equation} \label{eq:we2}
f_{\Theta^*}(\theta^*| \lambda,\delta, \xi) = \frac{\lambda e^{-\lambda (\delta\theta^*- \xi) }}{1-e^{-2 \pi \lambda}}. 
\end{equation}
To see that the wrapped exponential is not ICID and ICO we need to prove that does not exist a $\lambda^*$ such that 
\begin{equation} \label{eq:we3}
\frac{\lambda e^{-\lambda (\delta\theta^*- \xi) }}{1-e^{-2 \pi \lambda}}=\frac{\lambda^* e^{-\lambda^* (\theta^*\text{ mod } 2 \pi)}}{1-e^{-2 \pi \lambda^*}}.
\end{equation}
The modulus on the right side of \eqref{eq:we3} is needed because the definition of the wrapped exponential given in equation \eqref{eq:we} assumes that the circular variable belongs to $\mathbb{D}$ while $\Theta^*\in \delta(\mathbb{D}+\xi)$. It is not easy to solve \eqref{eq:we3} for $\lambda^*$ but we can simplify the demonstration by noticing that all the distributions can be expressed as the  product of a normalization constant and a kernel. If two density are equal, their kernel must be identical. Then it is sufficient to prove that it does not exist $\lambda^*$ such that the following relation is true: 
\begin{equation} \label{eq:we4}
e^{-\lambda \delta\theta^* }= e^{-\lambda^*  (\theta^* \text{ mod } 2 \pi)}.
\end{equation}
If we can prove that \eqref{eq:we4} is not true for a particular choice of $\xi$ and $\delta$ we have shown that the wrapped exponential is not ICID and ICO. We consider $\delta=1$ and $\xi>0$ and we evaluate the relation \eqref{eq:we4} at $\Theta^*= 2 \pi$ and $\Theta^*= \xi$, note that for the particular set of parameters we chose, the points $2 \pi$ and $\xi$ belong to the domain of $\Theta^*$. We have 
$%\begin{align}
\lambda 2 \pi = 0$ if $\Theta= 2 \pi,$ and
$\lambda \xi = \lambda^* \xi$, if $ \Theta= \xi. 
$%\end{align} 
Then \eqref{eq:we4} is true only if $\lambda=\lambda^*=0$ but  $\lambda>0$  and then it follows that \eqref{eq:we4} is never verified for a valid value of $\lambda$. This proves that  the wrapped exponential is not ICID and ICO. \\
Note that using Preposition \ref{prep:1} we obtain the density in equation \eqref{eq:we2}, that  is the density of the \emph{invariant wrapped exponential}.  
\subsubsection*{Wrapped Weibull}
\cite{Sarma2011} propose the wrapped Weibull:
%\begin{equation}
%f_{\Theta}(\theta| \mu,\sigma^2) =\sum_{k : \theta+2 \pi k>\mu}^{\infty} \frac{(2 \pi \sigma^2)^{-\frac{1}{2}}}{(\theta+2 \pi k- \mu)} e^{-\frac{(\log(\theta+2 \pi k )-\mu)^2}{2 \sigma^2}}
%\end{equation}
%and
\begin{equation} \label{eq:weib}
f_{\Theta}(\theta|\lambda) =\sum_{k =0 }^{\infty} \lambda (\theta+2 \pi k)^{\lambda-1}e^{-(\theta+ 2 \pi k)^{\lambda}}.
\end{equation}
To show that \eqref{eq:weib} does not verify ICID and ICO,  we use the characteristic function: %of  $\Theta$: 
\begin{equation} \label{eq:wei1}
\varphi_{\theta}(p) = i \sum_{k \in \mathbb{Z}_{odd}^+}\frac{ p^{k}}{k!}\left( 1+\frac{k}{\lambda} \right)- \sum_{k \in \mathbb{Z}_{even}^+}\frac{p^{k}}{k!}\left( 1+\frac{k}{\lambda}\right),
\end{equation}
where $\mathbb{Z}_{odd}^+$ and $\mathbb{Z}_{even}^+$ are respectively the even (zero included) and odd integer numbers. If $\Theta$ and $\Theta^*$ have  the same distribution their characteristic function must be of the same functional form. Since $\Theta^*$ is a linear transformation of $\Theta$ its characteristic function is
\begin{align} 
\varphi_{\theta^*}(p) &=e^{i \delta \xi p } \varphi_{\theta}(\delta p) = \cos(\delta \xi p )\varphi_{\theta}(\delta p) +i \sin(\delta \xi p )\varphi_{\theta}(\delta p) \label{eq:char}\\&
=\sin(\delta \xi p )\sum_{k \in \mathbb{Z}_{odd}^+}\frac{ p^{k}}{k!}\left( 1+\frac{k}{\lambda} \right)-   \cos(\delta \xi p )   \sum_{k \in \mathbb{Z}_{even}^+}\frac{(\delta p)^{k}}{k!}\left( 1+\frac{k}{\lambda}\right)      \\&+i \left(  \cos(\delta \xi p )    \sum_{k \in \mathbb{Z}_{odd}^+}\frac{p^{k}}{k!}\left( 1+\frac{k}{\lambda} \right)-  \sin(\delta \xi p ) \sum_{k \in \mathbb{Z}_{even}^+}\frac{ (\delta p)^{k}}{k!}\left( 1+\frac{k}{\lambda}\right)   \right)\label{eq:wei2}
\end{align}
The real and imaginary parts of \eqref{eq:wei1} and \eqref{eq:wei2} differ unless $\delta=1$ and $\xi=0$, and then the wrapped Weibull is not ICID and ICO

\subsubsection*{Wrapped L\'evy}
Following \cite{fisher1996}, if a  circular random variable $\Theta$ has pdf 
\begin{equation} \label{eq:wl}
f_{\Theta}( \theta|\mu, \sigma^2) = \sum_{k: \theta \frac{l}{2 \pi}+2 k l-\mu>0 }^{\infty}  \sqrt{\frac{\sigma^2}{2 \pi}}\frac{ e^{-\frac{\sigma^2}{2\left(\theta \frac{l}{2 \pi}+2 k l-\mu\right) }}}{\left(\theta \frac{l}{2 \pi}+2 k l-\mu\right)^{\frac{3}{2}}}, \mu \in \mathbb{R}, \sigma^2 > 0
\end{equation}
it is said to be distributed as a wrapped  L\'evy. \\
Here again to prove that \eqref{eq:wl} does not hold the ICID and ICO properties we exploit the characteristic function of $\Theta$ and $\Theta^*$ that are respectively
$%\begin{align}
\varphi_{\theta}(p) = e^{i \mu p -\sqrt{-2 i\sigma^2 p }}$ and  
$\varphi_{\theta^*}(p) = e^{i  p \delta (\mu+\xi)  -\sqrt{-2 i \sigma^2 \delta p }}.
$%\end{align}
The two characteristic functions  are of the same kind if the parameters of the distribution of $\Theta^*$ are $\mu^*= \delta (\mu+\xi)$ and $\sigma^{2*}=\sigma^2 \delta $ but if $\delta=-1$  then $\sigma^{2*}$ is negative while, by definition of the Wrapped L\'evy distribution,  it must be positive, then the distribution is not ICID and ICO.
%\textbf{Controllare anche la Wrapped weighted exponential distribution }
\subsection{Discrete distributions} \label{sec:disc}
We are particularly interested in  discrete circular distributions since we did not found in the literature discrete circular distributions that satisfy the ICID and ICO properties (definitions \ref{def:def1} and \ref{def:def2}), apart the trivial case of the circular uniform. 

\subsubsection*{Discrete circular uniform}
The   discrete circular uniform  density \citep{Merdia1999} for $\Theta$ is of the form:
$%\begin{equation}
f_{\Theta}=\frac{1}{l}
$%\end{equation}
where $l$ is the number of l distinct points equally spaced in $\mathbb{D}$. Clearly $\Theta^*$ has the same exact density and then the discrete circular uniform  is ICID and ICO.\\

\subsubsection*{Wrapped Poisson}
A discrete circular variable $\Theta$ is said to follow a wrapped Poisson distribution \citep{Merdia1999} if  it has pdf
$%\begin{equation} \label{eq:wp}
f_{\Theta}(\theta| \lambda ) = 
\sum_{k=0  }^{\infty} \frac{\lambda^{\theta  \frac{l}{2 \pi}+ k l }e^{-\lambda}}{\left(\theta \frac{l}{2 \pi} + kl\right)!}. 
$%\end{equation}
%The random variable $ \Theta^*$ has pdf
%\begin{equation} \label{eq:wp2}
%f_{\Theta^*}(\theta^*| \lambda, \delta, \xi ) = 
%\sum_{k=0  }^{\infty} \frac{\lambda^{(\delta \theta^*- \xi   ) \frac{l}{2 \pi}+ kl }e^{-\lambda}}{\left((\delta \theta^*- \xi   ) \frac{l}{2 \pi}+ kl  \right)!},
%\end{equation}
The wrapped Poisson is not ICID and ICO and to see that we use its characteristic function. From \citep{Girija2014} we have that 
$%\begin{equation}
\varphi_{\theta}(p)= e^{  -\lambda \left( 1- e^{i \frac{2 \pi p}{l} }  \right)  } = e^{  -\lambda \left( 1- \cos  \frac{2 \pi p}{l} \right)  }    e^{ i \lambda    \sin \frac{2 \pi p}{l}   }
$%\end{equation}
and using equation \eqref{eq:char} we can compute
\begin{equation} \label{eq:chwp}
\varphi_{\theta^*}(p)=e^{i\delta \xi p} e^{  -\lambda \left( 1- e^{i\delta \frac{2 \pi p}{l} }  \right)  }=  e^{-\lambda \left( 1- \cos   \frac{2 \pi p}{l} \right)} e^{i \left(  \delta \xi p+  \lambda \sin   \delta \frac{2 \pi p}{l} \right)}.
\end{equation}
The imaginary part of the two characteristic function are equal only if $\delta=1$ and $\xi=0$, i.e. $\Theta^*=\Theta$,then it  follows that  the wrapped Poisson is not ICID and ICO. 
%\begin{equation}
%\lambda^{\delta \theta^* \frac{l}{2 \pi}} \sum_{k=0  }^{\infty} \frac{\lambda^{ kl }}{\left((\delta \theta^*- \xi   ) \frac{l}{2 \pi}+ kl  \right)!}= (\lambda^*)^{ \theta^* \frac{l}{2 \pi}} \sum_{k=0  }^{\infty} \frac{(\lambda^*)^{ kl }}{\left(\theta^*\frac{l}{2 \pi}+ kl  \right)!},
%\end{equation}
%note that, as we did for the wrapped exponential, we use the kernel of the density. Now suppose $\delta=1$ and let write $\lambda^*=c \lambda$. The term $c$

\subsubsection*{Wrapped geometric}
Another  example is  the wrapped geometric    proposed by \cite{Jayakumar2012}. The pdf of $\Theta$ is 
\begin{equation}\label{eq:wg}
	f_{\Theta}( \theta|\lambda) = \frac{\lambda (1-\lambda)^{\theta l/(2 \pi)}}{1-(1-\lambda)^l}.
\end{equation}
We can see that  \eqref{eq:wg} has not the two properties because it is a decreasing function in $\Theta$ and then its maximum value is reached always at $\Theta=0$ and it cannot satisfy the two properties given in definitions \ref{def:def1} and \ref{def:def2}.

\subsubsection*{Wrapped skew Laplace on integers}

The wrapped skew Laplace on  integers of \cite{Jayakumar2012} has pdf 
$%\begin{equation} \label{eq:wsl}
f_{\Theta}( \theta|d,q) = \sum_{k= 0 }^{\infty} \frac{(1-d)(1-q)}{1-dq} \left(d^{\theta \frac{l}{2 \pi}+2 k l} +q^{|-\theta \frac{l}{2 \pi}-2 k l|}\right), \,d,q \in (0,1)
$%\end{equation}
%\begin{equation} \label{eq:wsl}
%f_{\Theta}( \theta|p,q) = \frac{(1-p)(1-q)}{1-pq} \left( \frac{q^{l-\theta\frac{l}{2 \pi}}(1-p^l)+ p^{\theta\frac{l}{2 \pi}}(1-q^l)}{(1-p^l)(1-q^l)}\right), \, p,q \in (0,1)
%\end{equation}
In each interval $[l (k-1),l (k-1)+ (l-1) )$ the maximum value of the term inside the sum is reached at $l (k-1)$ that, over the circle, corresponds to the point $\Theta=0$. Here again, as with the wrapped geometric example, since the maximum value of the density is fixed on the point $\Theta=0$, the density cannot be ICID and ICO. 

%while  the one of  $\Theta^* $  is 
%\begin{equation} \label{eq:gemetric}
%f_{\Theta^*}( \theta^*|\lambda) = \frac{\lambda (1-\lambda)^{(\delta \theta^*-\xi )  l/(2 \pi)}}{1-(1-\lambda)^l},
%\end{equation}
%Equation \eqref{eq:gemetric} is no more a  wrapped geometric pmf.\\
\subsubsection*{Wrapped binomial}
The wrapped binomial of  \cite{Girija2014} has pdf 
\begin{equation} \label{eq:wbin}
f_{\Theta}( \theta|q,n) = \sum_{k=0}^{k: n-(\theta \frac{l}{2 \pi}+k l)\geq 0 } 
\left(
\begin{array}{c}
n\\
\theta \frac{l}{2 \pi}+k l
\end{array}
\right) q^{\theta \frac{l}{2 \pi}+k l} (1-q)^{n-(\theta \frac{l}{2 \pi}+k l)},
\end{equation}
and characteristic function
\begin{equation} \label{eq:cwb1}
\varphi_{\Theta}(p) = (1-q+qe^{ip})^n.
\end{equation}
while the characteristic function of $\Theta^*$ is 
\begin{equation}\label{eq:cwb2}
\varphi_{\Theta^*}(p) = e^{i \delta \xi p }(1-q+qe^{i\delta p})^n.
\end{equation}
Equations \eqref{eq:cwb1} and \eqref{eq:cwb2} are not of the same functional form and then the wrapped binomial is not ICID and ICO. 

\section{A new discrete circular distribution: The Invariant wrapped Poisson 
% {la Generalized Poisson distribution esiste e non \'e legata a quella che introduciamo. Invariant suona sicuramente meglio}
} \label{sec:pois}

%\textbf{(Il nome generalized wrapped Poisson \`e brutto, se avete altre sono felice di combiarlo )}
 {In this Section,  relying on  Proposition  \ref{prep:1}, we build a new discrete circular distribution that is ICID and ICO.  Among the \emph{non-invariant} distributions discussed in the previous Section, we chose the wrapped Poisson  for several reasons. 
First of all the resulting distribution has a wide range of possible shapes, see  Figure \ref{fig:denspois}, ranging from a highly concentrated distribution to the discrete circular uniform as limit case (Section \ref{sec:circmom}). Second its  trigonometric moments can be written in closed form, allowing for an easy computation of   the circular concentration, directional mean and circular skewness (again see   Section \ref{sec:circmom}).  Third, parameters estimates can be easily estimated in a Bayesian framework by introducing some latent variables (see Section \ref{sec:bayest}). }

\subsection{The Invariant wrapped Poisson distribution}.

%As shown in the previous Section, the wrapped Poisson distribution is not ICID and ICO but we can use Proposition \ref{prep:1} to  build a version of this distribution that  {holds} the desired  {properties}. 
Let $X \sim Pois(\lambda)$ and let $\Theta^*= (2 \pi X/l) \text{ mod } 2 \pi $ be its circular counterpart, assuming $l$ distinct values over the unit circle. $\Theta^*$  is then distributed as a wrapped Poisson $WPois(\lambda)$. Proposition \ref{prep:1} tells us that the random variable $\Theta = \delta(\Theta^*+\xi)$ has pdf that verifies the two desired properties.  

% {This is the core of the work. More details are needed. The relation between equations (7) and (8) should be clearly defined.}

We include $\delta$ and $\xi$ in the set of characterizing parameters without restricting $\Theta$ to $\mathbb{D}$ (i.e we do not apply  the module $2 \pi$ operation in the transformation $\Theta = \delta(\Theta^*+\xi)$) then the distribution we propose,  
the \emph{ Invariant Wrapped Poisson distribution},  has pdf: 
\begin{equation} \label{eq:pmf}
f_{\Theta}(\theta| \lambda, \delta, \xi ) = 
\sum_{k=0  }^{\infty} \frac{\lambda^{(\delta \theta- \xi   ) l/(2 \pi)+ kl }e^{-\lambda}}{((\delta \theta- \xi   ) l/(2 \pi)+ kl)!}, \,  \Theta \in \delta(\mathbb{D}+\xi)
%\begin{array}{ll}
%\theta \in [\xi,2\pi+\xi)  &\text{ if } \delta=1\\
%\,\theta \in (-2\pi-\xi,-\xi] &\text{ if } \delta=-1
%\end{array}
%= \sum_{k=0  }^{\infty} \frac{\lambda^{(\delta \theta- \xi   ) \text{ mod } (2 \pi) l/(2 \pi)+k }e^{-\lambda}}{((\delta \theta- \xi   ) \text{ mod } (2 \pi) l/(2 \pi)+ k)!} .
\end{equation}
%\textbf{Ci vogliono i passaggi non si capisce bene come usi la proposition 1}
This formulation of the IWP produces a distribution with domain dependent on both $\delta$ and $\xi$ this can be problematic for model fitting. On the other hand we mentioned above that we are interested in restricting  $\Theta$ to $[0,2 \pi)$. One advantage of this choice is that by applying the modulus operation when computing $\Theta = \delta(\Theta^*+\xi)$ we obtain a random variable with domain independent from the characterizing parameters. In this setting the IWP pdf becomes:
%\begin{equation} \label{eq:pmf5}
%f_{\Theta}(\theta| \lambda, \delta, \xi ) = 
%\sum_{k:(\delta \theta- \xi   ) l/(2 \pi)+ kl \geq 0 }^{\infty} \frac{\lambda^{(\delta \theta- \xi   )  l/(2 \pi)+ kl }e^{-\lambda}}{((\delta \theta- \xi   ) l/(2 \pi)+ kl)!}, \, \theta \in [0,2 \pi)
%\end{equation}
%Equation \eqref{eq:pmf5} can be also written as 
\begin{equation} \label{eq:pmf1}
f_{\Theta}(\theta| \lambda, \delta, \xi ) = 
\sum_{k=0  }^{\infty} \frac{\lambda^{(\delta \theta- \xi   ) \text{ mod } (2 \pi) l/(2 \pi)+ kl }e^{-\lambda}}{((\delta \theta- \xi   )\text{ mod } (2 \pi) l/(2 \pi)+ kl)!}
\end{equation}
%where in this case the domain of $k$ is independent on $\theta$, $\delta$ and $\xi$.
In general we say that if a random variable $\Theta$ has pdf \eqref{eq:pmf1} (or \eqref{eq:pmf}) it has  an \emph{invariant wrapped  Poisson distribution with $l$ values ($IWP_l$)}. If $\delta=1$ and $\xi=0$ then the  $IWP_l$ reduces to the standard  $WP_l$.  %\todo[size=\tiny,color=green!40]{Va chiarito quello che ci siamo detti oggi rispetto agli special cases}.
%The representation of the pdf given in \eqref{eq:pmf} makes easy to find the mean trigonometric moments of $\Theta$ while \eqref{eq:pmf1} simplifies  inference in particular in a Bayesian framework, see Section \ref{sec:bayest}.  

\begin{figure}[t!]
	\centering
	{\subfloat[$\lambda=0.5$, $\xi=0$]{\includegraphics[trim=80 55 55 55,clip,scale=0.38]{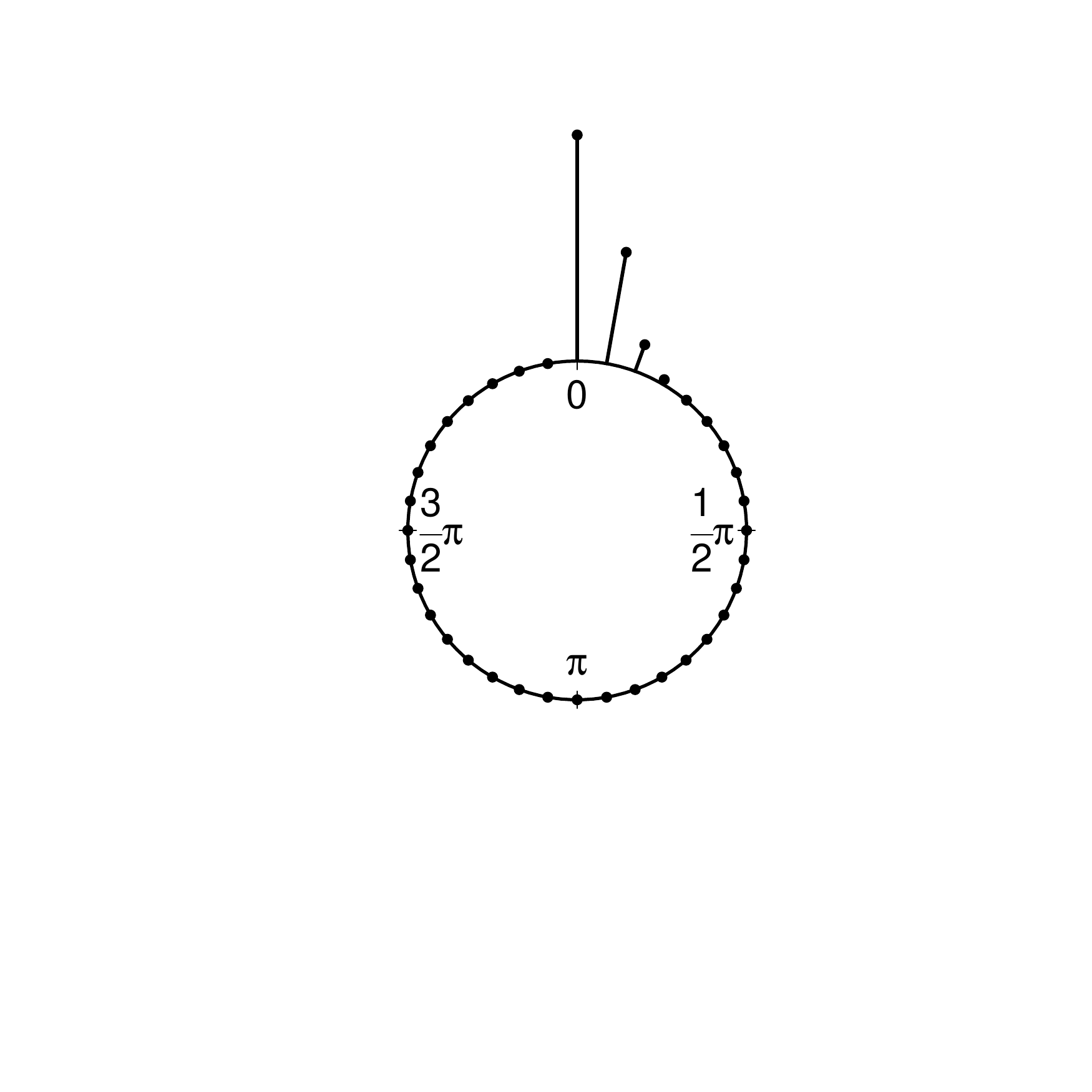}}}
	{\subfloat[$\lambda=1$, $\xi=0$]{\includegraphics[trim=80 55 55 55,clip,scale=0.38]{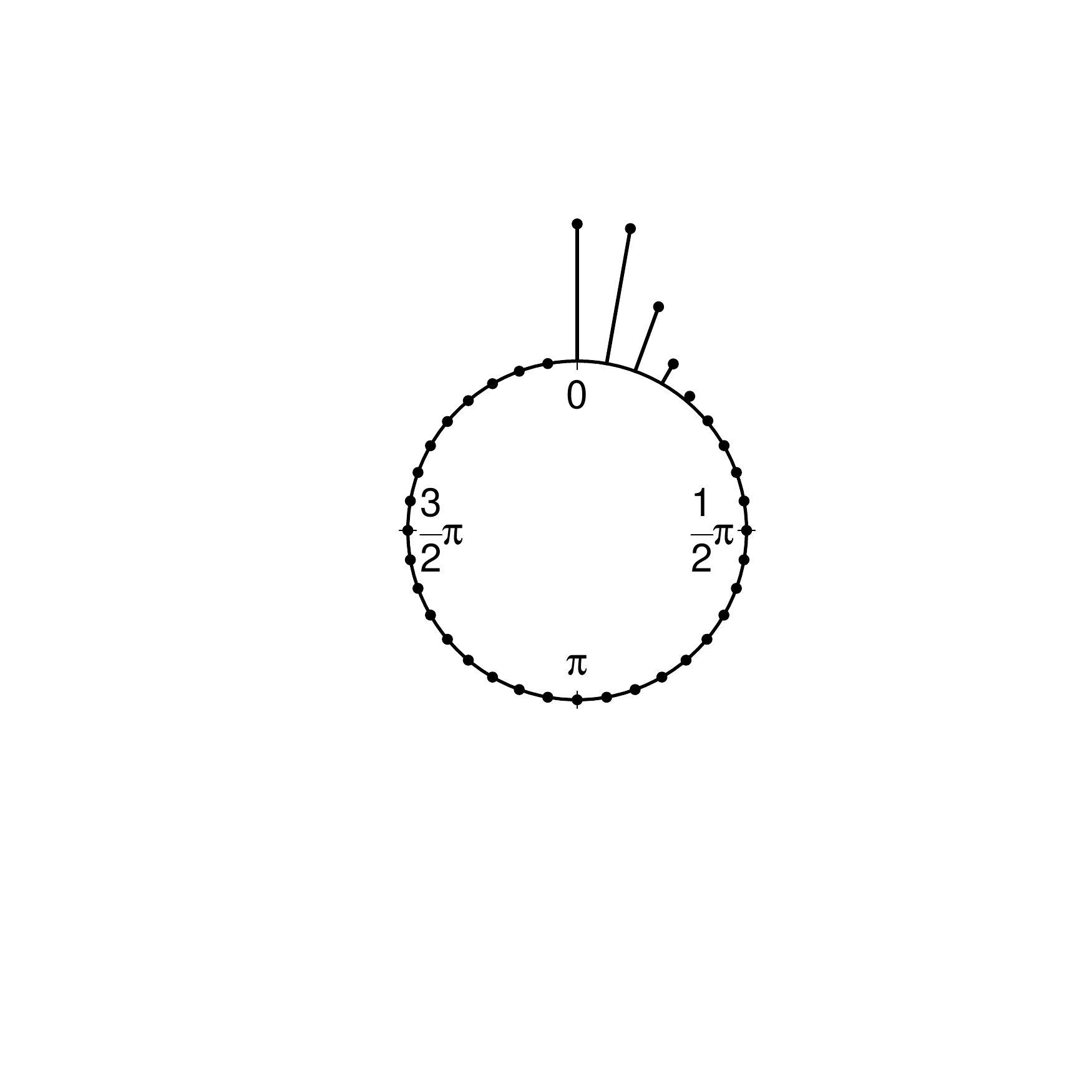}}}
	{\subfloat[$\lambda=2$, $\xi= 35\frac{2 \pi}{36}   $]{\includegraphics[trim=80 55 55 55,clip,scale=0.38]{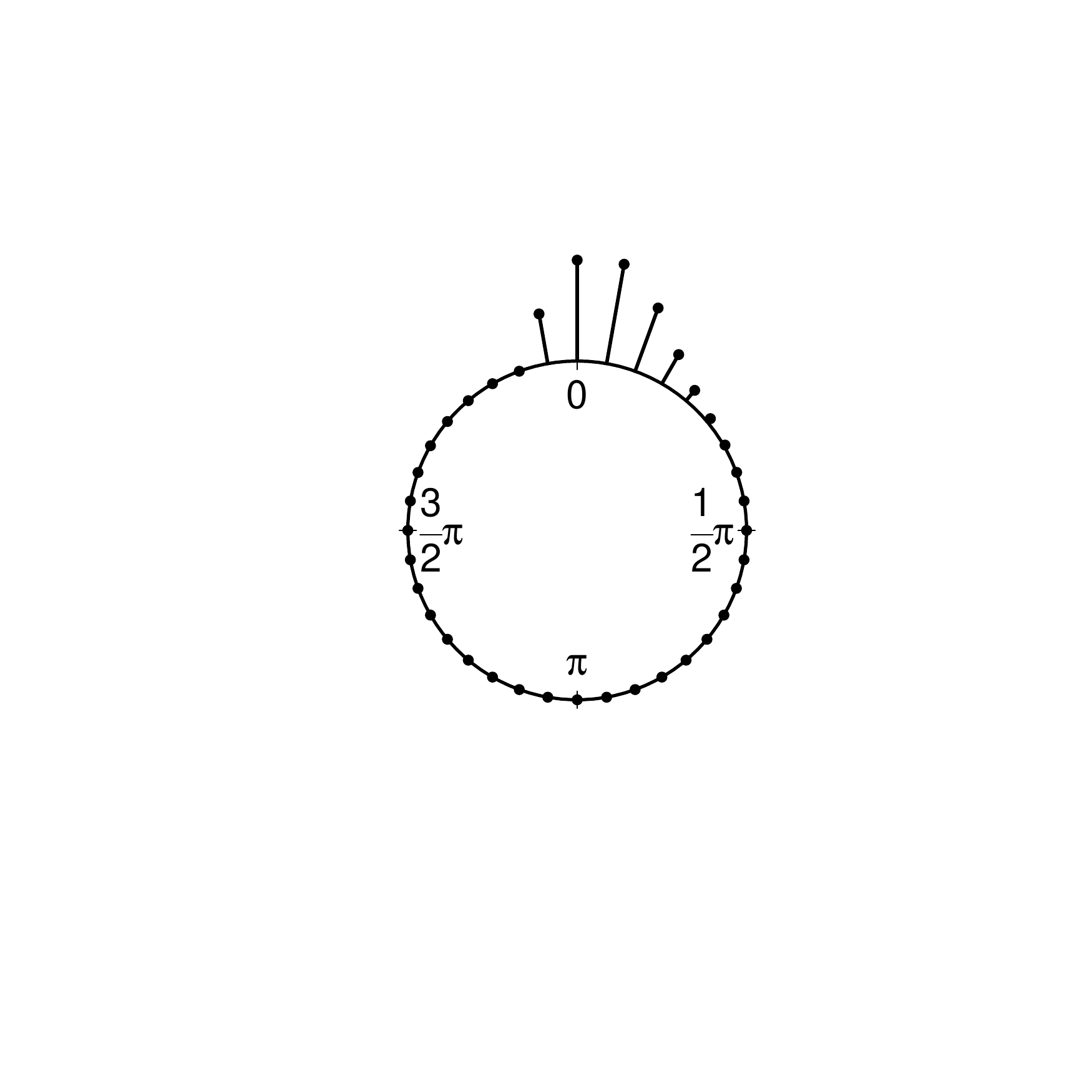}}}\\
	{\subfloat[$\lambda=5$, $\xi=32 \frac{2 \pi}{36}$]{\includegraphics[trim=80 55 55 55,clip,scale=0.38]{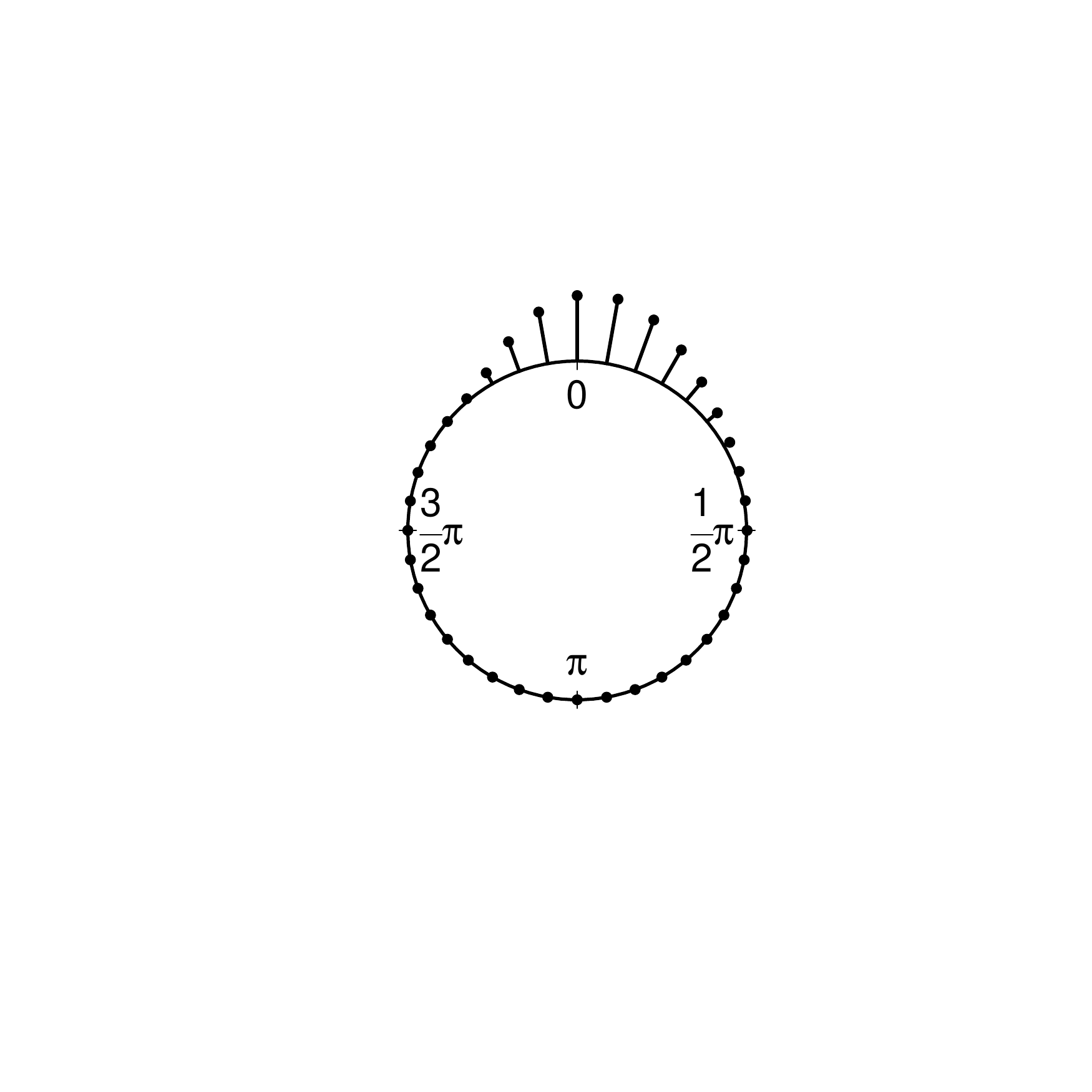}}}
	{\subfloat[$\lambda=20$, $\xi=16 \frac{2 \pi}{36}$]{\includegraphics[trim=80 55 55 55,clip,scale=0.38]{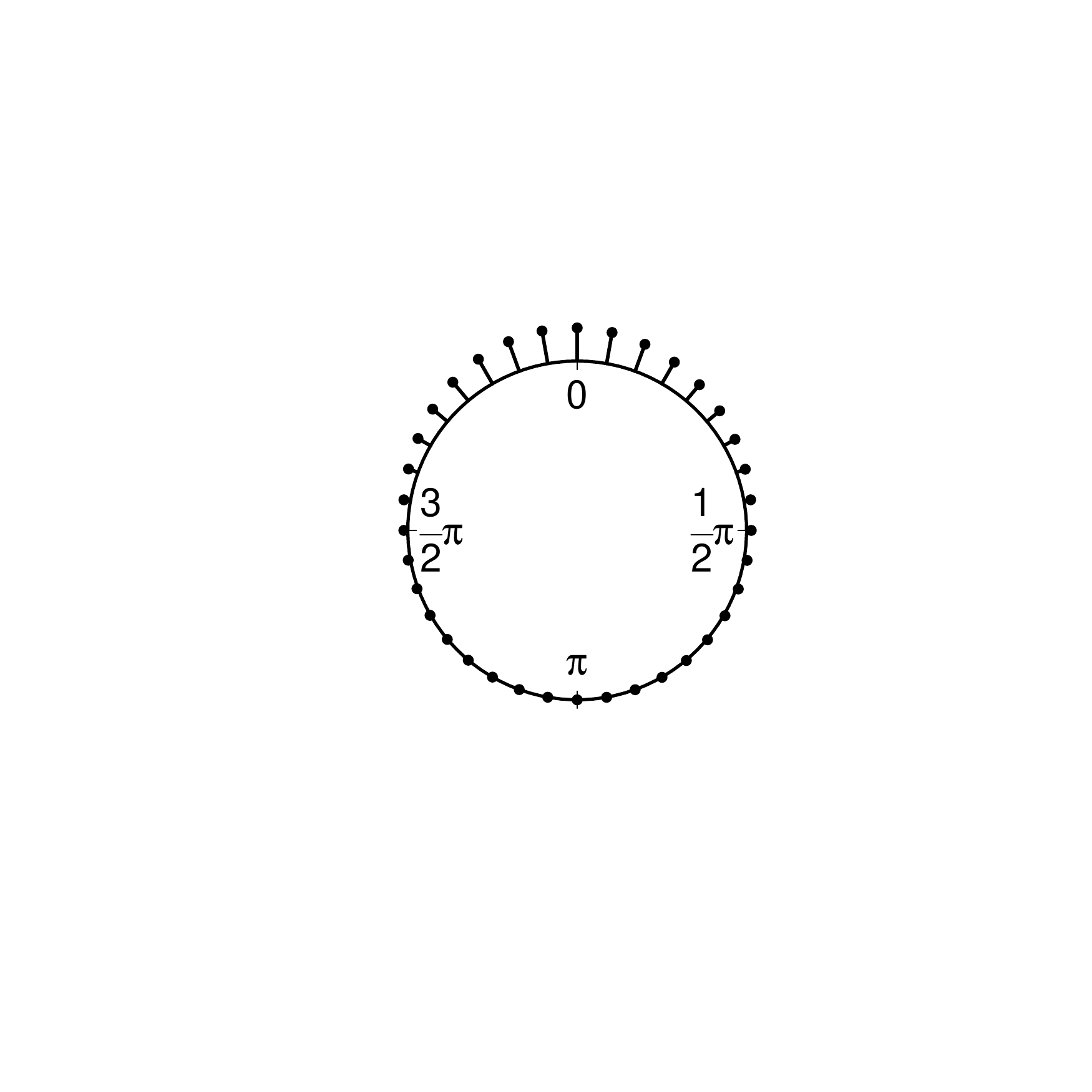}}}
	{\subfloat[$\lambda=500$, $\xi=0$]{\includegraphics[trim=80 55 55 55,clip,scale=0.38]{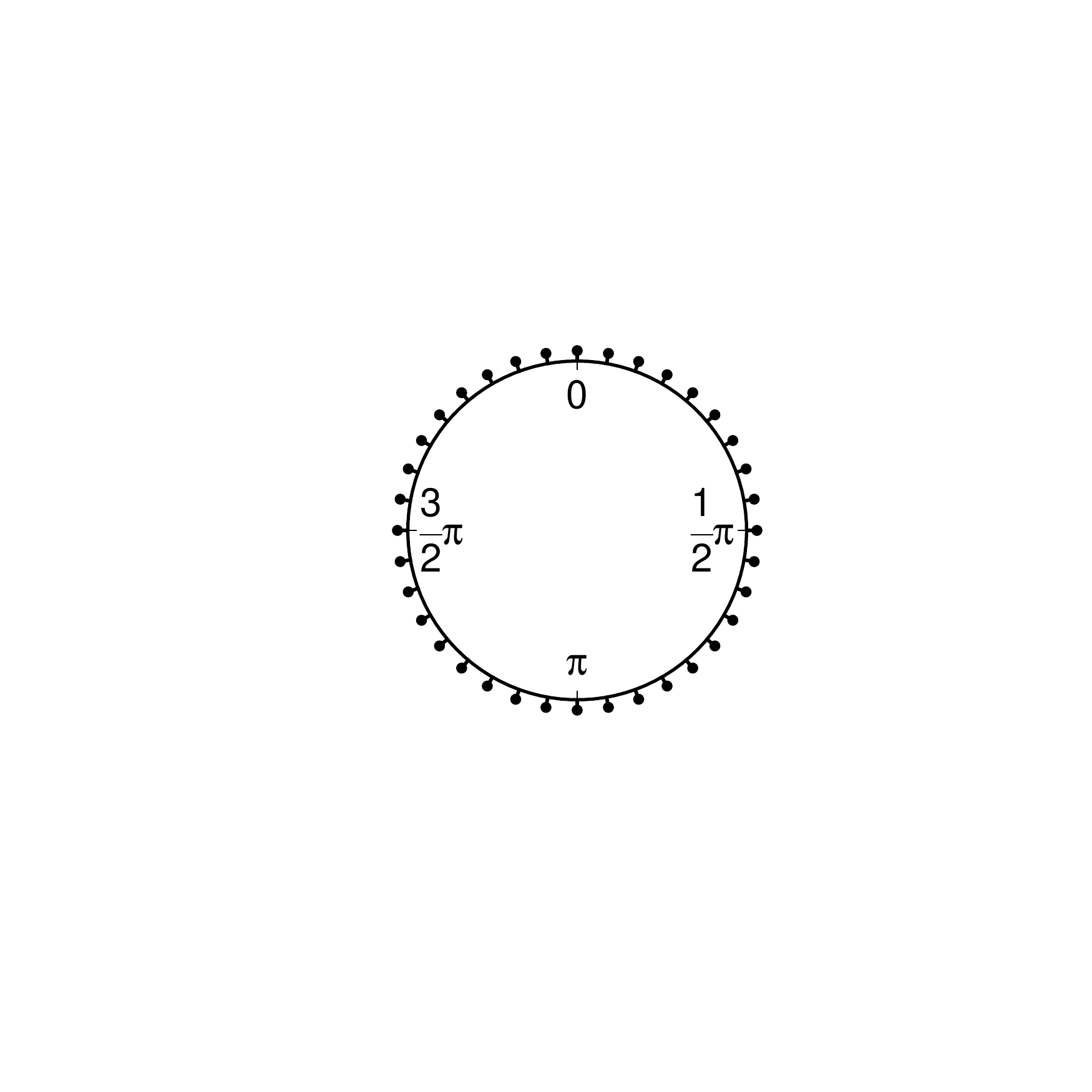}}}
	\caption{Examples of $IWP_{36}$ with $\delta=1$ and different values of  $\lambda$ and $\xi$.} \label{fig:denspois}
\end{figure} 

% {Non ho capito come si lega il paragrafo successivo al resto} \textbf{\`e da spostare}
\subsection{Trigonometric moments } \label{sec:circmom}

The trigonometric moments, defined as $\alpha_p = E \cos p\Theta$ and $\beta_p = E \sin p\Theta$, are quantities used to inform about many  features of  circular distributions, for example the mean direction,  $\mu = \mbox{atan}^*\frac{\beta_1}{\alpha_1}$ (the function $\text{atan}^*$ is the modified $\text{atan}$ function, see for example \cite{Jammalamadaka2001}), the circular concentration, $c = \sqrt{\alpha_1^2+\beta_1^2}$, and the circular skewness, $s= \frac{c\sin(\mu_2-2\mu)}{(1-c)^{\frac{3}{2}}}$  where $\mu_2= \mbox{atan}^*\frac{\beta_2}{\alpha_2}$  \citep[see][]{Merdia1999}.  $\alpha_p$ and $\beta_p$ are related to the characteristic function of the circular distribution, defined only on integer $p$ \citep[see again ][]{Merdia1999}:  $\varphi_{\theta}(p)=E(e^{ip\Theta})= \alpha_p+ i \beta_p$.

 {The characteristic function in equation \eqref{eq:chwp}, is the characteristic function of the random variable that arises applying  Preposition \ref{prep:1} to the wrapped Poisson distribution, then 
\eqref{eq:chwp} is the characteristic function of the invariant version of the wrapped Poisson, i.e. the characteristic function of the IWP. Hence} 
$%\begin{equation}
\varphi_{\theta}(p)=  e^{-\lambda \left( 1- \cos   \frac{2 \pi p}{l} \right)} e^{i \left(  \delta \xi p+  \lambda \sin   \delta \frac{2 \pi p}{l} \right)}
$%\end{equation}
and since $ e^{i \left(  \delta \xi p+  \lambda \sin   \delta \frac{2 \pi p}{l} \right)}= \cos \left(  \delta \xi p+  \lambda \sin   \delta \frac{2 \pi p}{l}  \right)+i \sin  \left(  \delta \xi p+  \lambda \sin   \delta \frac{2 \pi p}{l}  \right)$ we have
$%\begin{align}
\alpha_p = e^{-\lambda \left( 1- \cos   \frac{2 \pi p}{l} \right)}$ $\cos \left(  \delta \xi p+  \lambda \sin   \delta \frac{2 \pi p}{l}  \right)$, 
$\beta_p = e^{-\lambda \left( 1- \cos   \frac{2 \pi p}{l} \right)}\sin \left(  \delta \xi p+  \lambda \sin   \delta \frac{2 \pi p}{l}  \right).
$%\end{align}
The mean direction, the circular concentration and the circular skewness are respectively 
\begin{align}
\mu =&  \delta \xi+  \lambda \sin   \delta \frac{2 \pi}{l}, \label{eq:circmean}\\
c=& e^{-\lambda \left( 1- \cos   \frac{2 \pi }{l} \right)},\label{eq:circconc}\\ 
s =&    \frac{c \sin \left( - 2 \lambda \sin(\frac{\delta 2 \pi  }{l}) \left( 1-\cos(\frac{2 \pi  }{l})  \right)  \right)}{(1-c)^{\frac{3}{4}}}. 
\end{align}

%\subsection{Parameter interpretability} \label{sec:parint} \todo[size=\tiny,color=pink]{nsecondo me qui non ci vuoleun paragrafo a parte, invece la lettura dei parametri va lasciata a seguire della loro definizione}
% {Sinceramente non so che aggiungere. Il pezzetto che ho aggiunto non sono manco sicuro che si capisca, ma non so bene come dirlo e sinceramente non so neanche se \`e interessante.}
The circular concentration is a function of $\lambda$ and $l$, its maximum value, $c=1$, is reached at $\lambda=0$ and it is inversely proportional to $\lambda$. As we expected, $c$ does not depend on $\delta$ and $\xi$. Note that
as $\lambda \rightarrow \infty$ the concentration goes to zero and then the $IWP_l$ becomes a discrete uniform distribution.  \\
 {The directional mean is function of all three parameters and note that if we change $\delta$ the directional mean shift to its specular circular value with respect to initial direction, 
e.g.  if with $\delta=1$ we have $\mu=d$, then with $\delta=-1$ we have $\mu=l-d$. \\
Note that since the circular concentration depends only on $\lambda$ while the directional mean is a function of all the three parameters,  for a given value of $c$, i.e. a given value of $\lambda$,  we can change the directional mean and let  $c$ unchanged by  modifying the value of $\xi$. Hence the IWP circular mean and concentration are independent, breaking the dependence between  directional mean and circular concentration that the wrapped Poisson,   inherits  from the dependence between the mean and the variance of the Poisson. }

The skewness depends on both $\lambda$ and $\delta$ with a  complex functional form and we can not say much about its value, apart that $s \rightarrow 0$ as $\lambda \rightarrow \infty$ since $c \rightarrow 0$. It is of interest to understand how the sign of $s$  changes. Since $c \in [0,1]$ the sign of $s$ depends only on  $\sin \left( - 2 \lambda \sin(\frac{\delta 2 \pi  }{l}) \left( 1-\cos(\frac{2 \pi  }{l})  \right)\right) $ then 
let 
$
h_r = \frac{ r \pi}{2\sin(\frac{2 \pi}{l})\left(1-\cos(\frac{2 \pi}{l})\right)}, \, r=0,1,\dots,
$
and
consider the intervals  
\begin{equation} \label{eq:int}
h_r  \leq\lambda<h_{r+1}.
\end{equation} 
%
%
%
%Note that 
%\begin{align}
%\sin \left( - 2 \lambda \sin(\frac{\delta 2 \pi  }{l}) \left( 1-\cos(\frac{2 \pi  }{l})\right)\right)  \in [0,1]& \, if \\
%\sin \left( - 2 \lambda \sin(\frac{\delta 2 \pi  }{l}) \left( 1-\cos(\frac{2 \pi  }{l})\right)\right)  \in [0,1]& \, if 
%\end{align}
%
%
%
%
%If $h_r\leq\lambda< h_{r+1}$ then $r \pi\leq \sin \left( - 2 \lambda \sin(\frac{\delta 2 \pi  }{l}) \left( 1-\cos(\frac{2 \pi  }{l})\right)\right)<(r+1)\pi$
Inside the intervals in equation \eqref{eq:int}
the sign of $\sin \left( - 2 \lambda \sin(\frac{\delta 2 \pi  }{l}) \left( 1-\cos(\frac{2 \pi  }{l})\right)\right) $ depends only on $\delta$. Precisely if  $r$ is even the sign of $s$ is positive if $\delta=-1$ and negative otherwise, while if $r$ is odd $s \leq 0$ if $\delta=-1$ and $s > 0$ otherwise. The sign of the skewness depends on both $\lambda$ and $\delta$.    
Remark that  at the right limit  of the intervals \eqref{eq:int} the circular concentration becomes
\begin{equation} \label{eq:d}
\exp\left\{-\frac{(r+1) \pi}{2 \sin(\frac{2 \pi}{l})}\right\}.
\end{equation}
\eqref{eq:d} is inversely proportional to $r$ and it decreases with $l$ as soon as $l\geq 4$.  For $l$  large enough (indicatively $l\geq 20$), at the right limit of the first interval ($h_0  \leq\lambda<h_{1}$) the circular concentration \eqref{eq:d} is smaller than 0.007, the invariant wrapped Poisson starts to become indistinguishable from the discrete uniform and the skewness  of the distribution is no more a useful information. Then if $l\geq 20$, in the range of values of $\lambda$ where $s$ gives information on the shape of the distribution, the sign of the skewness depends only on $\delta$.

\subsection{Random number generation} \label{sec:gen}

We have two different, and easy, ways to generate a value from a $IWP_l(\lambda,\delta, \xi)$.\\
In the first method we   simulate $Z$ from a Poisson with parameter $\lambda$ and then we obtain the associated wrapped circular variable   $\Theta=\left(Z\frac{2 \pi}{l}\right) \text{ mod }2 \pi$ that is distributed as  a $WP_l(\lambda)$. The discrete circular random variable $\Theta^* = \delta(\Theta+\xi)$ is then distributed as a  $IWP_l(\lambda,\delta, \xi)$. \\
To introduce the second method, note that, if $Z \sim P(\lambda)$,  the density of the random variable $Y= \delta(Z\frac{2\pi}{l}+\xi)$ is 
$%\begin{equation} \label{eq:y}
f_{Y}(y|\lambda, \delta, \xi) = \frac{\lambda^{ (\delta Y-\xi)\frac{l }{2 \pi}}e^{-\lambda}}{\left((\delta Y-\xi)\frac{l}{2 \pi}\right)!},
$%\end{equation} 
and we can see that the pdf of $\Theta$, given in equation \eqref{eq:pmf},   arises by wrapping the random variable $Y$. Then if $Z $ is a drawn from a $P(\lambda)$, the random variable $\Theta=Y \text{ mod }2 \pi$, where $Y= \delta(Z\frac{2\pi}{l}+\xi)$,  is invariant wrapped Poisson distributed with parameter $\lambda$, $\delta$ and $\xi$.\\
The two  proposed methods can be applied to the generation of any invariant wrapped distribution, for example if we use the first method and we first generate $X  $ from a geometric distribution with parameter $\lambda$, then  $\Theta=\left(Z\frac{2 \pi}{l}\right) \text{ mod }2 \pi$ is distributed as a wrapped geometric and  $\Theta^* = \delta(\Theta+\xi)$ follows an invariant wrapped geometric with parameters $\lambda$, $\delta$ and $\xi$.

\subsection{Relation with the wrapped normal } \label{sec:relwn}

	It is well know that a Poisson distribution can be well approximated by a normal when the Poisson parameter when  is ``large' enough''. 
%	It is well know that a Poisson distribution can be well approximated by a normal distribution with mean and variance equal to the Poisson parameter when the latter is ``large' enough''. 
	Then, when variability is high,  if $Z \sim P(\lambda)$, we have $f_{Z}(z|\lambda) \approx \phi(z|\lambda-0.5,\lambda)$, where $\phi(\cdot|\cdot,\cdot)$ is the normal pdf. It follows that the density of  $Y =  \delta (\lambda Z  \frac{2 \pi}{l}+\xi)$  can be well approximated by 
	$\phi\left(y| \delta((\lambda-0.5)  \frac{2 \pi}{l}+\xi),   \lambda  \left(\frac{2 \pi}{l}\right)^2  \right)$. As described in Section \ref{sec:gen}, the invariant wrapped Poisson can be obtained by wrapping $Y$:
	$%\begin{equation}
	f_{\Theta}(\theta| \lambda,\delta, \xi)= \sum_{k=0}^{\infty}f_{Y}(\theta+2 \pi k| \lambda,\delta, \xi ),
	$%\end{equation}
	and since $f_{Y}(y| \lambda,\delta, \xi ) \approx \phi\left(y| \delta ((\lambda-0.5)  \frac{2 \pi}{l}+\xi),   \lambda  \left(\frac{2 \pi}{l}\right)^2  \right)= \phi\left(y| (\lambda-0.5) \frac{2 \pi}{l}+\xi,   \lambda  \left(\frac{2 \pi}{l}\right)^2  \right)$  we have 
	\begin{equation}\label{eq:nd}
	f_{\Theta}(\theta| \lambda,\delta, \xi)\approx \sum_{k=0}^{\infty}\phi\left(\theta+2 \pi k| (\lambda-0.5) \frac{2 \pi}{l}+\xi,   \lambda  \left(\frac{2 \pi}{l}\right)^2  \right).
	\end{equation}
	 The sum \eqref{eq:nd} is the density of a wrapped normal that well approximate  $f_{\Theta}(\theta| \lambda,\delta, \xi)$ when $\lambda$ is large enough  and then the mean of the normal density in \eqref{eq:nd} is far away from  0. When $l$ is relatively small a continuity correction must be added to account for the different type of distributions support.\\
Remark that when the ratio $\lambda/l^2$ is large the WN in \eqref{eq:nd} is indistinguishable from a circular uniform distribution \citep[see][for details]{Jona2013},in the same way  the IWP is indistinguishable from a circular discrete uniform. As an example, in Figure \ref{fig:supdist} the sup distance between the IWP, with $\delta=1,\xi=0$ and the corresponding approximating WN is reported for several values of $\lambda$ and $l$. It is clear that the approximation is accurate at the third decimal figure even with relatively small $\lambda/l$ ratios. However for the use we are envisioning in the sequel we consider ``acceptable'' the approximation when the sup distance between the two distributions is less than $0.001$.  In the next section we are going to  use this approximation to derive an efficient Markov chain Monte Carlo (MCMC) algorithm. \\ %in Section \ref{sec:bayest}.}\\

\begin{figure}
\begin{center}
\includegraphics[scale=0.3]{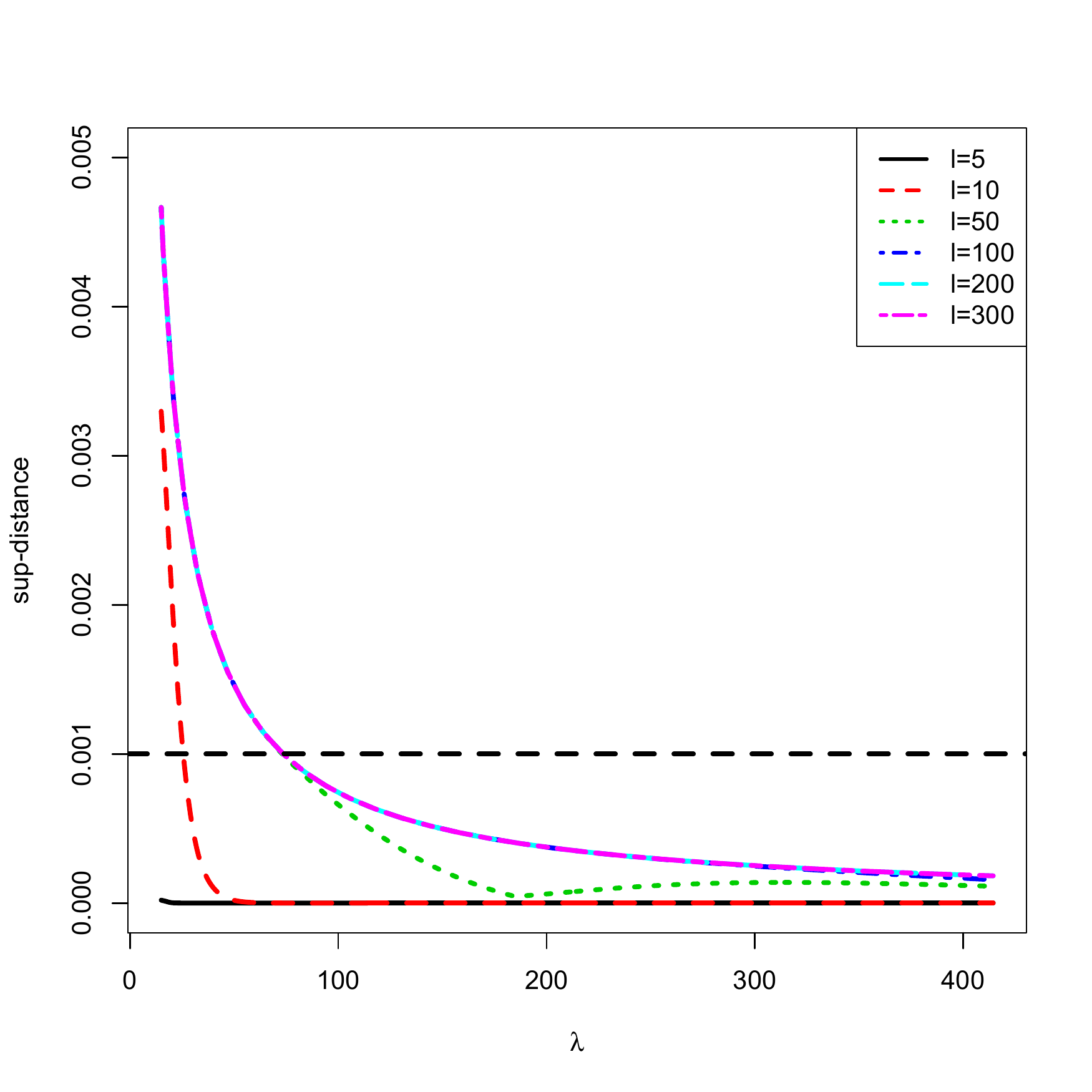}
\caption{Sup distance between Wrapped Normal approximation and Invariant Wrapped Poisson as a function of $\lambda$, dashed lines correspond to different values of $l=5,10,50,100,200,300$}\label{fig:supdist}
\end{center}
\end{figure}

\section{Statistical inference: Bayesian estimation} \label{sec:bayest}
%\todo[size=\tiny,color=green!40]{Pensa se vuoi anche mettere model selection}.

%\textbf{forse va specificato che qui usiamo la rappresentazione \eqref{eq:pmf1}}\\
In this section we illustrate how to obtain estimates of the IWP parameters and prediction from the same distribution in a Bayesian framework.  Notice that
wrapped distributions have convenient inferential properties  because, by introducing a latent variable in a data augmentation approach \citep[see for example][]{coles98}, we can linearize the circular variable avoiding the evaluation of the infinite sum involved in their definitions.   Here we  consider the random variable $X= (\delta \Theta- \xi   ) \text{ mod } (2 \pi) l/(2 \pi)+ Kl$, where in this case $K$ is a latent random variable. The joint pdf of $(\Theta,K)$ is the summand of \eqref{eq:pmf1} and a marginalization over $K$ gives $f_{\Theta}$: % {A cosa serve specificare la joint? Quali sono i vantaggi? \textbf{simuliamo sempre dalla joint}}
$%\begin{equation} \label{eq:pmf3}
f_{\Theta}(\theta|  \lambda, \delta, \xi) = 
\sum_{k=0  }^{\infty} f_{\Theta,K}(\theta,k|\lambda, \delta, \xi).
$%\end{equation}
%The IWP makes no exception and allows us to adopt a data-augmentation approach (\textbf{reference}) in simulating from it.
Now suppose to observe a sample of  $T$  values from a discrete circular variable $\Theta_t \sim IWP_l(\lambda, \delta,\xi), t=1,\dots, T$, we want to run posterior and predictive inference using this sample. The parameters are $\boldsymbol{\psi}=(\lambda,\delta,\xi)$ and to facilitate estimation we introduce the latent variables $K_t\in \mathbb{Z}$. Let   $X_t= (\delta \Theta_t- \xi   ) \text{ mod } (2 \pi) l/(2 \pi)+ K_tl$, $\mathbf{X}=\{X_t\}_{t=1}^T$, $\boldsymbol{\Theta}= \{\Theta_t\}_{t=1}^T$ and $\mathbf{K}= \{K_t\}_{t=1}^T$ we can write the (augmented) likelihood of $(\boldsymbol{\Theta}, \mathbf{K})$ as
$%\begin{equation}
f_{\boldsymbol{\Theta}, \mathbf{K}}(\boldsymbol{\theta},\mathbf{k}|\boldsymbol{\psi}) = \prod_{t=1}^T \frac{\lambda^{x_t}e^{-\lambda}}{x_t!}= \prod_{t=1}^T\frac{\lambda^{(\delta \theta_t- \xi   ) \text{ mod } (2 \pi) l/(2 \pi)+ k_tl }e^{-\lambda}}{((\delta \theta_t- \xi   )\text{ mod } (2 \pi) l/(2 \pi)+ k_tl)!}. 
$%\end{equation}
As usual in the Bayesian framework, to estimate the posterior distribution of  $\mathbf{k},\boldsymbol{\psi}|\boldsymbol{\theta}$ we obtain posterior samples using a Markov chain Monte Carlo (MCMC) algorithm. From \eqref{eq:circmean} it is clear that all parameters are involved in the evaluation of the  directional mean. This implies that  a component-wise simulation scheme based on parameters full conditional distributions, could be highly inefficient.  A typical solution  to this problem is to block sample \citep{finley2013} the parameters. Then by  choosing the appropriate prior distributions   we can update $\{\lambda, \delta,\xi\}$ jointly. We assume independence between $\lambda$ and $(\delta, \xi)$ and    as prior distribution for $\lambda$  a Gamma, $\lambda \sim G(v,m)$ while no specific requirement are needed for modeling prior information on $\delta$ and $\xi$, indeed the distribution of $\xi$ can be used to model  prior information on the sign of the skewness if $l\geq 20$, see Section \ref{sec:circmom}.\\
 Using composition  sample, to draw a sample from $\lambda, \delta, \xi|\mathbf{k},\boldsymbol{\theta}$ we can first sample from $\delta, \xi|\mathbf{k},\boldsymbol{\theta}$ and then from $\lambda| \delta, \xi,\mathbf{k},\boldsymbol{\theta}$. Using the linear variables $X_t$s and the   conjugacy between the Poisson distribution and the Gamma we can find the distribution of $\delta, \xi|\mathbf{k},\boldsymbol{\theta}$:
\begin{equation} \label{eq:ful1}
f_{\delta, \xi|\mathbf{k},\boldsymbol{\theta}}(\delta, \xi|\mathbf{k},\boldsymbol{\theta}) \propto \frac{\Gamma( \sum_{t=1}^T x_t+v )}{   (m+T)^{\sum_{t=1}^T x_t+v}    \prod_{t=1}^T\left( x_t \right)! }f_{\delta,\xi}(\delta,\xi)
\end{equation}
We  evaluate \eqref{eq:ful1} for all the possible combinations of  values  $(\delta,\xi)$, and then we can sample from the discrete distribution $f_{\delta, \xi|\mathbf{k},\boldsymbol{\theta}}(\delta, \xi|\mathbf{k},\boldsymbol{\theta}) $. The full conditional of $\lambda$ is a $G(\sum_{t=1}^Tx_t+v, T+m)$ and  then we can easily simulate from it. To conclude the  MCMC specification  the latent variable $k_t$ can be simulated with a Metropolis step. \\
We occasionally experienced slow mixing determined by the separate sampling of $\lambda$ and $k_t$'s.  Block sampling of $\lambda, \mathbf{k}$ is computationally unfeasible.   Hence, to solve the problem, we need a different strategy: we are going to use the  density approximation described in  section \ref{sec:relwn} to ``appropriately'' limiting $\lambda$ and $k_t$'s ranges and we adopt a different formalization for the winding numbers  to be used in the MCMC algorithm. \\
%The main responsible for  the erratic behaveiour of the chains is $\lambda$. When this parameter assumes a ``large'' value the entire simulation moves far away the distribution of interest as samples become similar to does from a circular uniform distribution. We can avoid this behavior by ``appropriately'' limiting $\lambda$ and $k_t$'s ranges.\\  
We know that the  IWP can be well approximated  by the WN when the ratio $\lambda/l$ is relatively large, and then both distributions are closed to the circular uniform (discrete and continuos respectively).  \cite{Jona2013} show that, for a given sample size we can find,   through simulation, the  maximum value of the wrapped normal variance, $\sigma_{\max}^2$ at which   a wrapped normal and a circular uniform are not significantly different. Let $\lambda_{max} =  \left(\frac{l}{2 \pi}\right)^2\sigma_{\max}^2$, then  \eqref{eq:nd} with $\lambda= \lambda_{max}$ becomes 
$%\begin{equation}
f_{\Theta}(\theta| \lambda_{max},\delta, \xi)\approx \sum_{k=0}^{\infty}\phi\left(\theta+2 \pi k| (\lambda_{max}-0.5)  \frac{2 \pi}{l}+\xi,  \sigma_{\max}^2 \right).
$%\end{equation}
the sum describes a wrapped normal density evaluated at  $\sigma^2\approx \sigma_{\max}^2$ implying that the WN is almost a circular uniform.  Again the  $IWP_l(\lambda_{max}, \delta, \xi)$ and the WN  are indistinguishable from a (discrete and continuos) circular uniform and we are not interested in  $\lambda>\lambda_{\max}$. Then we can truncate  $\lambda$ domain  to $(0,\lambda_{\max})$. The constrained domain of $\lambda$ and again a result from \citet{Jona2013}, let us define a limited  range for $k_t$. We define   $k_{max}$, such that for any value of $\lambda$, the equation 
\begin{equation} \label{eq:aa2}
\sum_{k_t=0}^{k_{max}}\frac{\lambda^{(\delta \theta_t- \xi   ) \text{ mod } (2 \pi) l/(2 \pi)+ k_tl }e^{-\lambda}}{((\delta \theta_t- \xi   )\text{ mod } (2 \pi) l/(2 \pi)+ k_tl)!}. 
\end{equation}
 well approximate the density of the IWP.  \\
\citet{Jona2013} show that we can approximate the wrapped normal density if we let $k_t \in [\hat{k}_{low}, \hat{k}_{up}]$, with  $\hat{k}_{low}>-\frac{3 \sigma}{2 \pi}+\frac{\mu}{2 \pi}+\frac{1}{2}$ and  $\hat{k}_{up}<\frac{3 \sigma}{2 \pi}+\frac{\mu}{2 \pi}-\frac{1}{2}$ where $\mu$ is the mean  of the WN. Notice that both $\hat{k}_{low},\hat{k}_{up}$ depend on the WN parameters that in turn depend directly on $\lambda$ in \eqref{eq:nd}. Here  we are looking for an upper bound for $k_t$ and then we can focus on  $\hat{k}_{up}$. %Note that the value of $\hat{k}_{up}$ increases with $\mu$ and $\sigma$ and since in the approximation given in \eqref{eq:nd},  the WN's parameters are directly proportional to $\lambda$, $\hat{k}_{up}$ increases also with it. 
If we set $\lambda=\lambda_{max}$ we obtain the desired upper bound for $k_t$ ($k_{max}$)%the $\hat{k}_u$ associate with $\lambda_{\max}$, i.e. $k_{max}$, we are sure that this is the higher value we need to have a good approximation for all $\lambda \in (0, \lambda_{max})$. 
Then let $k_{\max}=  \lceil \frac{3 \sqrt{\lambda_{max}}}{l}+\frac{\lambda_{max}}{l}-\frac{1}{2} \rceil$ and using the result in \cite{Jona2013}, 
 the probability mass captured by the approximation \eqref{eq:aa2} is  
\begin{equation} \label{eq:aa3}
\sum_{\theta_t \in \mathbb{D}}\sum_{k_t=0}^{k_{max}}\frac{\lambda^{(\delta \theta_t- \xi   ) \text{ mod } (2 \pi) l/(2 \pi)+ k_tl }e^{-\lambda}}{((\delta \theta_t- \xi   )\text{ mod } (2 \pi) l/(2 \pi)+ k_tl)!}>0.997,\, \forall \lambda \in [0, \lambda_{max}]. 
\end{equation}
and \eqref{eq:aa3} becomes close to 1 when $\lambda << \lambda_{max}$. 
%\textbf{Vorrei scrivere che siccome abbiamo trovato $k_{max}$ usando $\lambda_{max}$  e quindi quando l'apprissimazione tra IWP e WN \`e migliore, allora la \eqref{eq:aa3} \`e vera. Ma mi sembra ripetitiva come cosa....che dici?}	 
Now  we decompose the winding numbers in two components:  $k_t = (\tilde{k}+\bar{k}_t)\text{ mod } k_{max} $, one ($\bar{k}_t \in [0, k_{max}]$) specific to the $t^{th}$ observation and one $(\tilde{k}\in [0, k_{max}])$ common to all $k_t$s.   Let $\bar{\mathbf{k}}= \{\bar{k}_t\}_{t=1}^T$,   the pmf of $\boldsymbol{\theta},\tilde{k},\bar{\mathbf{k}}$, based on the approximated density \eqref{eq:aa2},  is 
$%\begin{equation}
\prod_{t=1}^T\frac{1}{k_{max}}\frac{\lambda^{(\delta \theta_t- \xi   ) \text{ mod } (2 \pi) l/(2 \pi)+ (\tilde{k}+\bar{k}_t)\text{ mod } (k_{max}) l }e^{-\lambda}}{(\delta \theta_t- \xi   )\text{ mod } (2 \pi) l/(2 \pi)+ (\tilde{k}+\bar{k}_t)\text{ mod } (k_{max})l)!}=\prod_{t=1}^T \frac{\lambda^{x_t}e^{-\lambda}}{x_t!k_{max}}
$%\end{equation}
where $x_t=(\delta \theta_t- \xi   )\text{ mod } (2 \pi) l/(2 \pi)+ (\tilde{k}+\bar{k}_t)\text{ mod } (k_{max}) l$. In this setting we can block sample $(\{\lambda,\delta,\xi\},\tilde{k})$ by first sampling from  $f_{\delta, \xi,\tilde{k}|\bar{\mathbf{k}},\boldsymbol{\theta}}(\delta, \xi,\tilde{k}|\bar{\mathbf{k}},\boldsymbol{\theta}) $ and then from the full conditional of $\lambda$.  Remark that we can direct sample from $f_{\delta, \xi,\tilde{k}|\bar{\mathbf{k}},\boldsymbol{\theta}}(\delta, \xi,\tilde{k}|\bar{\mathbf{k}},\boldsymbol{\theta}) $ only because $\tilde{k}$ has a limited domain. The functional form of $f_{\delta, \xi,\tilde{k}|\bar{\mathbf{k}},\boldsymbol{\theta}}(\delta, \xi,\tilde{k}|\bar{\mathbf{k}},\boldsymbol{\theta})$ is the same as in \eqref{eq:ful1} but now it must be seen as a function of $(\delta, \xi,\tilde{k})$ and it has to be evaluated for all the $2lk_{max}$ possible values of $(\delta, \xi,\tilde{k})$.  As a final remark observe that $\tilde{k}$ and $\bar{\mathbf{k}}$ are not identifiable but we can always transform them back to the associated $k_t$s that are identifiable.
%
%
%We have that 
%\begin{equation} \label{eq:ful12}
%f_{\delta, \xi,\tilde{k}|\bar{\mathbf{k}},\boldsymbol{\theta}}(\delta, \xi,\tilde{k}|\bar{\mathbf{k}},\boldsymbol{\theta})  \propto \frac{\Gamma( \sum_{t=1}^T x_t+v )}{   (m+T)^{\sum_{t=1}^T x_t+v}    \prod_{t=1}^T\left( x_t \right)! }f_{\delta,\xi}(\delta,\xi),
%\end{equation}
%that is the same of \eqref{eq:ful1}. 
 The full conditional of $\lambda$ is $G(\sum_{t=1}^T+v, T+m)I(0,\lambda_{\max})$. The latent variable $\bar{k}_t$ can be updated using a Gibbs step since its full conditional is proportional to  
$
\frac{\lambda^{(\delta \theta_t- \xi   ) \text{ mod } (2 \pi) l/(2 \pi)+ (\tilde{k}+\bar{k}_t)\text{ mod } (k_{max}) l }e^{-\lambda}}{(\delta \theta_t- \xi   )\text{ mod } (2 \pi) l/(2 \pi)+ (\tilde{k}+\bar{k}_t)\text{ mod } (k_{max})l )!}
$
and has to be evaluated only for a finite set of  values.

\section{ {Numerical e}xamples}

\subsection{ {Artificial data}} \label{sec:smalsim}
\begin{figure}[t!]
	\centering
	{\subfloat[]{\includegraphics[trim=80 55 55 55,clip,scale=0.45]{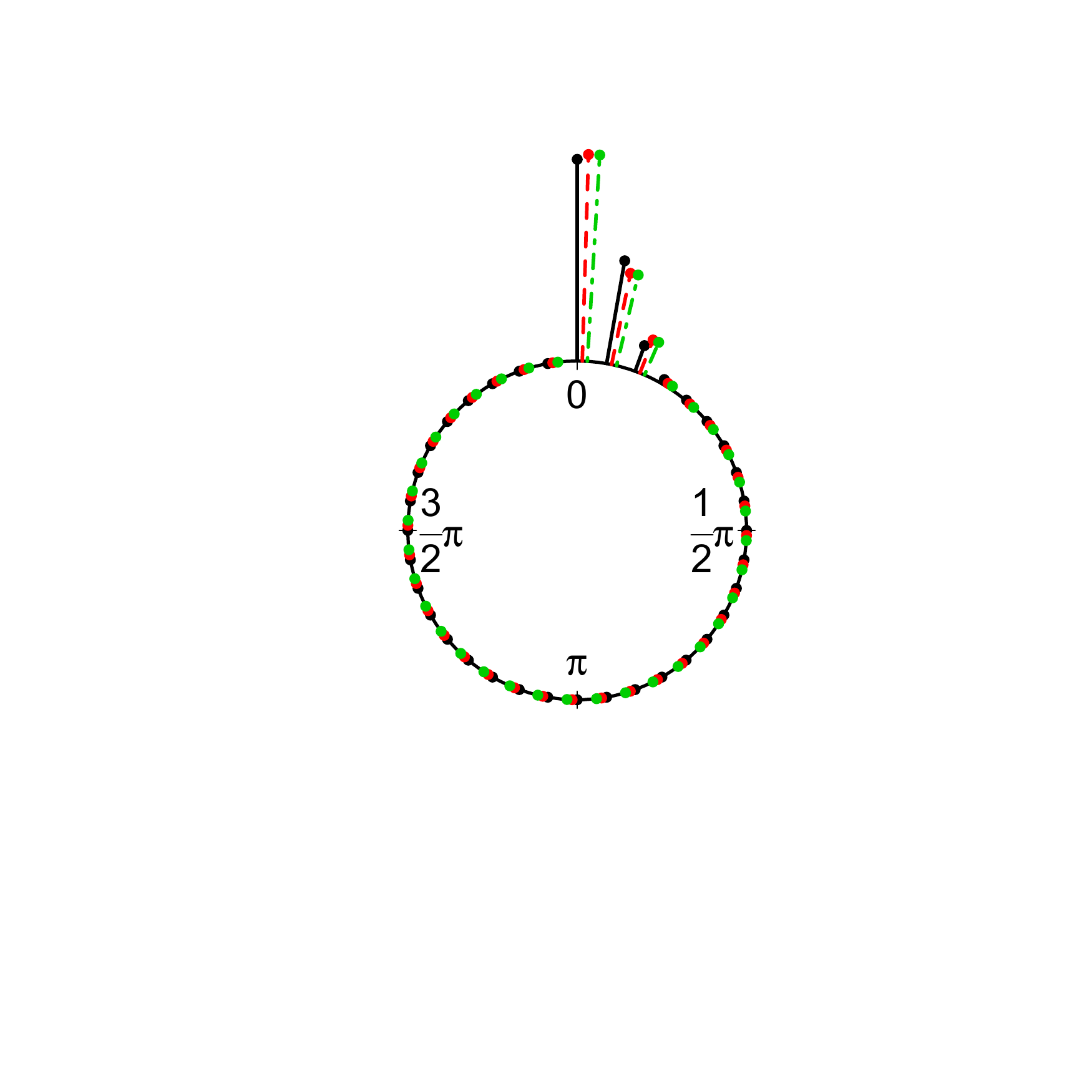}}}
	{\subfloat[]{\includegraphics[trim=80 55 55 55,clip,scale=0.45]{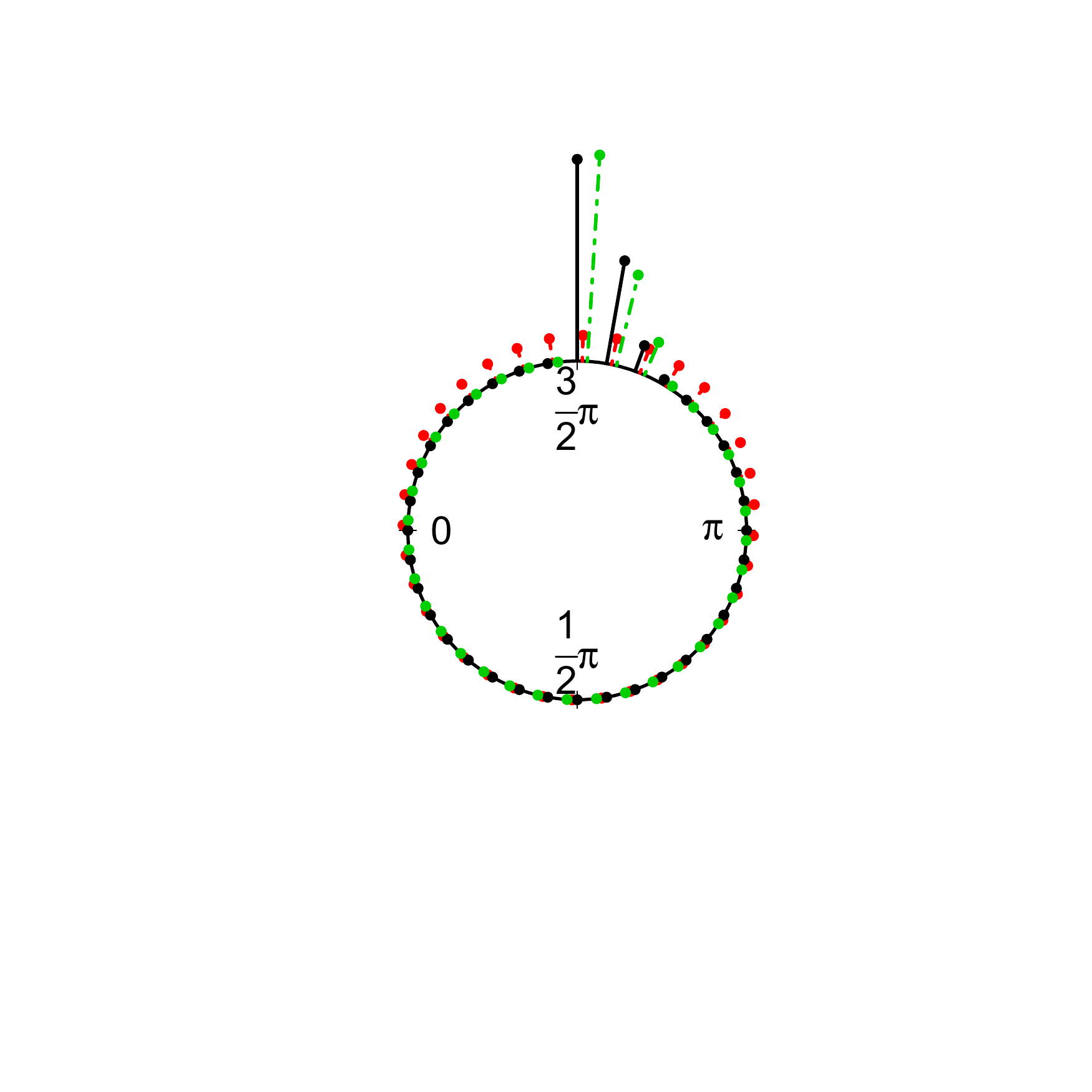}}}
\caption{Circular barplot of simulated (solid), posterior from WP (dashed) and posterior from IWP (dash-dots) values (a) example 1, (b) example 2:  same data as in (a) but with different reference system } \label{fig:denssim}
\end{figure}

%\begin{figure}[t!]
%	\centering
%	{\subfloat[]{\includegraphics[trim=80 55 55 55,clip,scale=0.38]{simWP1}}}
%	{\subfloat[]{\includegraphics[trim=80 55 55 55,clip,scale=0.38]{simWP11}}}
%	{\subfloat[]{\includegraphics[trim=80 55 55 55,clip,scale=0.38]{simWP12}}}
%	\\
%	{\subfloat[]{\includegraphics[trim=80 55 55 55,clip,scale=0.38]{simPP2}}}
%	{\subfloat[]{\includegraphics[trim=80 55 55 55,clip,scale=0.38]{simWP2}}}
%	{\subfloat[]{\includegraphics[trim=80 55 55 55,clip,scale=0.38]{simPP22}}}
%	\caption{Circular barplot of simulated values (first column (a) (d)), posterior predictive density using the Wrapped Poisson (second column (b) (e)) and using the Invariant Wrapped Poisson (third column (c) (f))  } \label{fig:denssim}
%\end{figure}

%
%
%\begin{figure}[t!]
%	\centering
%	{\subfloat[]{\includegraphics[trim=80 55 55 55,clip,scale=0.38]{SIM1}}}
%	{\subfloat[]{\includegraphics[trim=80 55 55 55,clip,scale=0.38]{simp1}}}
%	{\subfloat[]{\includegraphics[trim=80 55 55 55,clip,scale=0.38]{simp2}}}
%	 \\
%	{\subfloat[]{\includegraphics[trim=80 55 55 55,clip,scale=0.38]{SIM2}}}
%	{\subfloat[]{\includegraphics[trim=80 55 55 55,clip,scale=0.38]{simp3}}}
%	{\subfloat[]{\includegraphics[trim=80 55 55 55,clip,scale=0.38]{simp4}}}
%	\caption{Circular barplot of simulated values (first column (a) (d)), posterior predictive density using the Wrapped Poisson (second column (b) (e)) and using the Invariant Wrapped Poisson (third column (c) (f))  } \label{fig:denssim}
%\end{figure}

By means of a simulated example, we show the consequences of the lack of ICO and ICID properties on  {model} inference.  {Let us} simulate data from a Wrapped Poisson with parameter $\lambda$, i.e. an IWP with $\delta=1$ and $\xi=0$. Then  we estimate on one hand the posterior distribution of $\lambda$ using the WP and on the other hand the set of parameters $\{\lambda, \delta,\xi\}$ using the IWP.  We compute the posterior estimate of the mean direction, circular concentration and the  predictive density. Then we modify the orthogonal system and we show how the posterior estimates change. \\
We simulate $T=$200 realizations from  a wrapped Poisson  with 36 values over the unit circle and parameter $\lambda=0.5$  ($\Theta \sim WP_{36}(\lambda)$) and we choose a weakly informative prior for $\lambda$ i.e. $\lambda \sim G(1,0.0005)I(0,500)$ that is roughly uniform between $[0,500]$. A circular barplot of the simulated data is given in Figure  \ref{fig:denssim} (a).  
Using Equations \eqref{eq:circmean} and \eqref{eq:circconc} we can compute the true directional mean and concentration that are 0.087 and 0.992.\\ We estimate the model using the approximated MCMC with 10000 iteration,   burnin 4000 and thinning 2, i.e.
we estimated the posterior distribution using 3000 samples.  Under a WP model the posterior mean estimate  of $\lambda$ is $0.519$, with 95\% credible interval (CI)   $[0.423, 0.627]$ and the posterior mean direction and concentration are  0.090 (CI [0.073, 0.109]\footnote{For the circular variable the 95\% CI is computed as the shorter arc that contains the 95\% of posterior samples }) and  0.992 (CI [0.991, 0.994]). The posterior predictive distribution, reported in Figure \ref{fig:denssim} (a), is close to the barplot of Figure  \ref{fig:denssim} (a).  Under the IWP model we obtain very similar results with $\lambda=0.519$ (CI $[0.423,0.627 ]$), while $\delta=1$ and $\xi=0$ along all simulations. The circular mean is 0.090 (CI [0.073, 0.109]) and the concentration 0.992 (CI [0.991, 0.994]).
\\
We then changed the orientation and the initial direction of the reference system. In Figure \ref{fig:denssim}  (b) the data are depicted in the new system. Data are obtained as $\Theta^*= (-(\Theta+\pi/2)) \text{  mod } 2 \pi$. Note that $\Theta^*$ can be seen as a draw from a $IWP_{36}(0.5,-1,\pi/2)$ with  directional mean  4.626 and unchanged circular concentration (0.992). Under the WP model the posterior mean estimate of $\lambda$ is  26.483 with  95\% CI  $[25.762, 27.200]$. The posterior directional mean is 4.599 (CI [4.474, 4.723]), fairly close to the ``true'' value, but the circular concentration is 0.669 (CI [0.662, 0.676]), that is seriously underestimated and then the predictive density, the dashed line in Figure \ref{fig:denssim}  (b), is less concentrated and starts to resemble the circular uniform.  Under the IWP model the estimate of $\lambda$   as well as its credible interval does not change with respect to the first example. Estimates of $\delta$ and $\xi$ are -1 and $\frac{\pi}{2}$ respectively in all simulations. The circular mean is 4.622 (CI [4.603, 4.639]) and concentration is 0.992 (CI [0.991, 0.994]).  Figure \ref{fig:denssim} (b) illustrate these results.

\subsection{Wind direction}

\begin{figure}[t!]
	\centering
	{\subfloat[]{\includegraphics[trim=80 55 55 55,clip,scale=0.48]{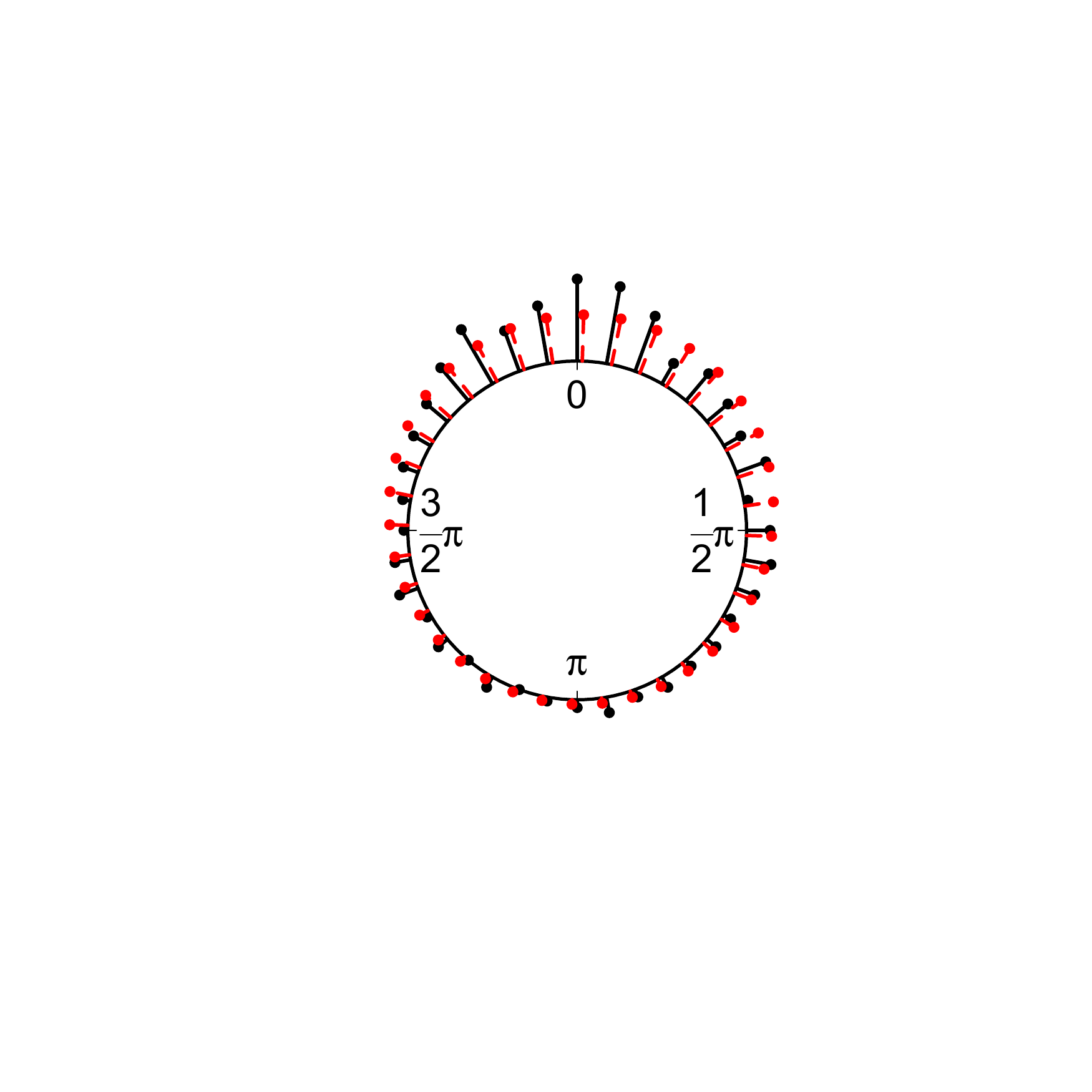}}}
	{\subfloat[]{\includegraphics[trim=80 55 55 55,clip,scale=0.48]{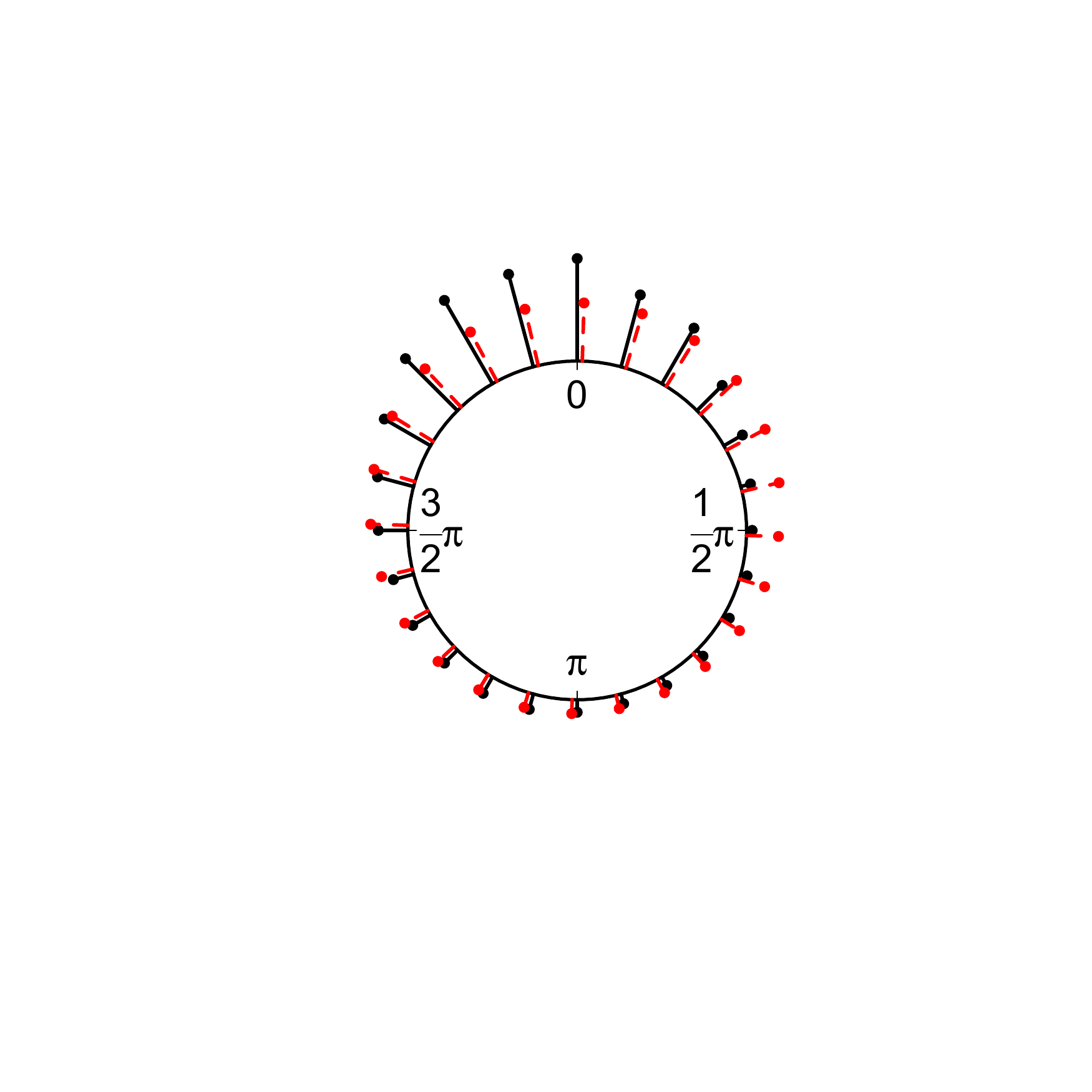}}} 
	\caption{Circular barplot of observed (solid) and predicted (dashed) values (a)  wind data  and (b) gun crime data } \label{fig:Wind}
\end{figure}

 {This section describes data on wind directions recorded on January 2000 at the monitoring station of Capo Palinuro. The monitoring station of Capo Palinuro (WMO code 16310) is one of the coastal stations managed by the Meteorological Service of the Italian Air Force. The station is located on the rocky cape of Capo Palinuro, in the town of Centola in the province of Salerno, South Italy. The cape goes for two kilometers in the Tyrrhenian Sea, between the Velia gulf and the Policastro gulf. Its geographical coordinates are 40$^\circ$01'31.06''N �15$^\circ$16'50.71''E GW and the station elevation is 185 meters above sea level. To the east of the station there are mountains following a North West� South East direction. Accordingly, it is plausible that the orography affects the observed measurements, leading to a unimodal distribution indicating a prevalent direction. }

 {Wind directions are monitored and routinely collected by several environmental agencies. 
Analyzed data come from reports prepared at the station and  provided by the National Center of Aeronautical Meteorology and Climatology (C.N.M.C.A.), special office of the Meteorological Service of the Italian Air Force.
% Records available cover a period from 1951 to 2010, for a total of 161.979 data. 
The  database includes date and time of registration, direction of the wind in degrees, with eight daily measurements (every three hours), i.e. we have 240 observations. The measuring instrument, anemometer, is placed away from obstacles and at a height of 10 meters above ground. A relevant issue with this recordings is that the measurement instrument asses wind directions on a discrete scale dividing the circle into ten-degree intervals ($l=36$). Accordingly, the need for a discrete circular distribution to proper model these data. The barplot of the available data is provided in Figure \ref{fig:Wind}(a).}%The SYNOP type provides data of the wind to the ground level with its direction (the average of the origin of the wind referred to the 10-minute period preceding the time of observation) and intensity (average of the speed, expressed in knots, recorded in the same time range used for the direction). The measuring instrument, called anemometer, is placed away from obstacles and at a height of 10 meters above the ground.}

%For the model estimate we take the data in the first month of the 2000, we have than 124 observation. Since the data has 36 points over the unit circle, as the simulate example, we use the same value of $s_2$ and $h$.  
%As for the simulated example in Section \ref{sec:smalsim}, here and in section \ref{sec:gun}, 
To get parameters estimates, we use the  approximated MCMC illustrated in section \ref{sec:bayest} with 10000 iterations, burnin 4000 and thinning 2.\\
%\textbf{Questo paragrafo \`e un po' un casino. Dove sono i grafici delle distribuzioni? Poi se $prob(\delta=1)>>prob(\delta=-1)$ perch�\`e dire che c'\`e incertezza sul segno di $s$? dovrebbe essere prevalentemente positiva. La moda sull'asse negativo deve essere molto pi\'u piccola di quella sull'asse positivo. In generale poi se le due mode sono comunque vicine a zero e la media a posteriori di $s$ ha un CI con lo zero dentro possiamo dire che la distribuzione \`e simmetrica.}
The marginal posterior distribution of $\delta$  returns a probability value of 0.23 for $\delta= -1$ and of 0.77  for $\delta= 1$ suggesting a negative  skewness $s$, since, as we noted in Section \ref{sec:pois}, with $l \geq 20$ the sign of the skewness depends only on $\delta$.  The correspondent posterior distribution of  $s$ is bimodal with one mode on the negative axis and  a smaller one on the positive one;  the posterior mean of the skewness is -0.192, the 95\% CI is    $[-0.198, 0.196]$ ( [-0.198, -0.187] conditioning to $\delta=1$ and [0.184, 0.197] conditioning to $\delta=-1$), suggesting that the predictive distribution is almost symmetric, see the predictive distribution in Figure \ref{fig:Wind}(a).  The posterior   95\% CI of $\lambda$ is [47.809,  59.442] with posterior mean  51.744 that is far away from 0 and then we expect a fairly dispersed predictive distribution. Indeed the posterior distribution of the circular concentration has mean value 0.456    and CI [0.405, 0.484] (recall that $c_{max}=1$). The posterior distribution of $\xi$ is negatively skewed  assigning zero probability below $\xi=12 \frac{2 \pi}{36}$ and above $\xi=25 \frac{2 \pi}{36}$,  its  modal value is $24 \frac{2 \pi}{36}$ reached with probability 0.123.  The posterior mean estimates of the directional mean is 0.164 (closed to the geographical North), with CI [6.248, 0.366].  To evaluate the goodness of fit of the model we use the average prediction error proposed in \citet{Jona2013} that is computed  as
$\frac{1}{B} \sum_{b=1}^B  \sum_{t=1}^n (1-\cos(\theta_t-\theta_t^b))$, where $B$ is the number of posterior samples, $T$ is the number of observations used to estimate the model, $\theta_t$ is the $t^{th}$ observation and $\theta_t^b$ is a posterior simulation of the  $t^{th}$ observation using as parameters the ones of the  $b^{th}$ posterior sample; in this example we obtain a value 0.780, indicating a reasonable fit given the large variability of the observations (maximum value for the ape is 2).

\subsection{Gun crime data}\label{sec:gun}

%\begin{figure}[t!]
%	\centering
%	{\subfloat[]{\includegraphics[trim=80 55 55 55,clip,scale=0.58]{Gun}}}
%	{\subfloat[]{\includegraphics[trim=80 55 55 55,clip,scale=0.58]{GunEst}}} 
%	\caption{Circular barplot (a) of the gun crime data  and (b) its posterior predictive distribution } \label{fig:Gun}
%\end{figure}

In a number of areas of police resource management, analysis of the relative frequency
of events occurring at different times of day plays an important role. Identifying times of
day when certain kinds of crime or disorder are more likely to occur offers a first step
towards prioritising resources and improved targeting.

The data set described in this section relates to reports of crimes committed with a gun in  Pittsburgh, Pennsylvania,  recorded over the period 1987-1998 \citep{cohen2006,Gill2010}. To avoid under- or over-reporting related to the time lag between the occurrence of the crime and the time of its reporting, a discrete scale is considered, i.e. data are collected on hourly basis.
We can think of the data as circular discrete data with 24 points over the unit circle. The barplot of the data is shown in Figure \ref{fig:Wind} (b), where the initial direction is midnight, and the sense of orientation is clockwise. The data contains information on 15831 crimes. An initial inspection of the data reveals an unimodal distribution, where frequency of reporting times peaks during the night, and falls off during the day.

The posterior  directional mean is 6.166 with 95\%  CI [6.138,  6.193] that on the 24 hours scale, corresponds  to 23:33 ([23:27,23:39]), close to midnight.   The posterior estimate of  $\lambda$ is  33.834 (CI [33.727, 33.941]) implying a large variability of the data ($\lambda>>0$), confirmed by the estimate of the circular concentration 0.316 with 95\% CI [0.315, 0.317] .  $\delta=-1$  with probability 1 and it follows that the posterior distribution of the skewness is non-zero over the positive axis, with a point estimate of 0.236 and 95\% CI [0.235, 0.236] . Again  the value of the skewness is small suggesting an almost symmetric  predictive density,  see Figure \ref{fig:Wind} (b), The posterior of $\xi$ is concentrated on $15 \frac{2 \pi}{24}$. In this example the average prediction error is 0.838 again a reasonable fit given the observed variability.

\section{Summary and concluding remarks}
In this paper we developed a method to build circular distribution invariant for changes in the reference system. In particular we discuss and analyze in details a new discrete circular distribution that helps the modeling of several phenomena. Simulated examples illustrated how misleading can be the use of not invariant distributions. Real world examples show that from the description of natural processes to the investigation of social phenomena using discrete circular data we  can benefit from an appropriate statistical modeling.\\ Future work will be focused on classification problems using hidden markov models based on discrete invariant circular distribution being the wind data our motivating example.  Further developments will involve the study of these type of distributions with dependent (in both time and space) phenomena.

\bibliography{all}

\begin{thebibliography}{}

\bibitem[Bulla {\em et~al.}(2012)Bulla, Lagona, Maruotti, and
  Picone]{Bulla2012}
Bulla, J., Lagona, F., Maruotti, A., and Picone, M. (2012).
\newblock {A multivariate hidden Markov model for the identification of sea
  regimes from incomplete skewed and circular time series}.
\newblock {\em Journal of Agricultural, Biological, and Environmental
  Statistics\/}, {\bf 17}(4), 544--567.

\bibitem[Cohen and Gorr(2006)Cohen and Gorr]{cohen2006}
Cohen, J. and Gorr, W. (2006).
\newblock {\em Examination of Crime Guns and Homicide in Pittsburgh,
  Pennsylvania, 1987-1998\/}.
\newblock Inter-university Consortium for Political and Social Research (ICPSR)
  {$[$}distributor{$]$}.

\bibitem[Coles(1998)Coles]{coles98}
Coles, S. (1998).
\newblock {Inference for Circular Distributions and Processes}.
\newblock {\em Statistics and Computing\/}, {\bf 8}(2), 105--113.

\bibitem[Coles and Casson(1998)Coles and Casson]{Coles1998}
Coles, S. and Casson, E. (1998).
\newblock Extreme value modelling of hurricane wind speeds.
\newblock {\em Structural Safety\/}, {\bf 20}(3), 283 -- 296.

\bibitem[Eckert {\em et~al.}(2008)Eckert, Moore, Dunn, van Buiten, Eckert, and
  Halpin]{Eckert2008}
Eckert, S.~A., Moore, J.~E., Dunn, D.~C., van Buiten, R.~S., Eckert, K.~L., and
  Halpin, P.~N. (2008).
\newblock Modeling loggerhead turtle movement in the mediterranean: Importance
  of body size and oceanography.
\newblock {\em Ecological Applications\/}, {\bf 18}(2), 290--308.

\bibitem[Finley {\em et~al.}(2013)Finley, Banerjee, and Gelfand]{finley2013}
Finley, A.~O., Banerjee, S., and Gelfand, A.~E. (2013).
\newblock {spBayes for Large Univariate and Multivariate Point-Referenced
  Spatio-Temporal Data Models}.
\newblock {\em Journal of Statistical Software (forthcoming)\/}.

\bibitem[Fisher(1996)Fisher]{fisher1996}
Fisher, N.~I. (1996).
\newblock {\em {Statistical Analysis of Circular Data}\/}.
\newblock Cambridge University Press, Cambridge.

\bibitem[Gill and Hangartner(2010)Gill and Hangartner]{Gill2010}
Gill, J. and Hangartner, D. (2010).
\newblock Circular data in political science and how to handle it.
\newblock {\em Political Analysis\/}, {\bf 18}(3), 316--336.

\bibitem[Girija {\em et~al.}(2014)Girija, Rao, and Srihari]{Girija2014}
Girija, S. V.~S., Rao, A. V.~D., and Srihari, G. V. L.~N. (2014).
\newblock On wrapped binomial model characteristics.
\newblock {\em Mathematics and Statistics\/}, {\bf 2}(7), 231 -- 234.

\bibitem[Jammalamadaka and Kozubowski(2004)Jammalamadaka and
  Kozubowski]{Rao2004}
Jammalamadaka, S.~R. and Kozubowski, T.~J. (2004).
\newblock New families of wrapped distributions for modeling skew circular
  data.
\newblock {\em Communications in Statistics - Theory and Methods\/}, {\bf
  33}(9), 2059--2074.

\bibitem[Jammalamadaka and SenGupta(2001)Jammalamadaka and
  SenGupta]{Jammalamadaka2001}
Jammalamadaka, S.~R. and SenGupta, A. (2001).
\newblock {\em {Topics in Circular Statistics}\/}.
\newblock World Scientific, Singapore.

\bibitem[Jayakumar and Jacob(2012)Jayakumar and Jacob]{Jayakumar2012}
Jayakumar, K. and Jacob, S. (2012).
\newblock Wrapped skew laplace distribution on integers:a new probability model
  for circular data.
\newblock {\em Open Journal of Statistics\/}, {\bf Vol.02No.01}, 9.

\bibitem[Jona~Lasinio {\em et~al.}(2012)Jona~Lasinio, Gelfand, and
  Jona~Lasinio]{Jona2013}
Jona~Lasinio, G., Gelfand, A., and Jona~Lasinio, M. (2012).
\newblock Spatial analysis of wave direction data using wrapped {G}aussian
  processes.
\newblock {\em Annals of Applied Statistics\/}, {\bf 6}(4), 1478--1498.

\bibitem[Lagona {\em et~al.}(2014)Lagona, Picone, Maruotti, and
  Cosoli]{Lagona2014}
Lagona, F., Picone, M., Maruotti, A., and Cosoli, S. (2014).
\newblock A hidden {M}arkov approach to the analysis of space-time
  environmental data with linear and circular components.
\newblock {\em Stochastic Environmental Research and Risk Assessment\/}, {\bf 29}, 397--407.

\bibitem[Lagona {\em et~al.}(2015)Lagona, Picone, and Maruotti]{lagona2015}
Lagona, F., Picone, M., and Maruotti, A. (2015).
\newblock A hidden markov model for the analysis of cylindrical time series.
\newblock {\em Environmetrics\/}, {\bf
  to appear}.

\bibitem[Langrock {\em et~al.}(2012)Langrock, King, Matthiopoulos, Thomas,
  Fortin, and Morales]{Langrock2012}
Langrock, R., King, R., Matthiopoulos, J., Thomas, L., Fortin, D., and Morales,
  J.~M. (2012).
\newblock Flexible and practical modeling of animal telemetry data: hidden
  {M}arkov models and extensions.
\newblock {\em Ecology\/}, {\bf 93}(11), 2336--2342.

\bibitem[Langrock {\em et~al.}(2014)Langrock, Hopcraft, Blackwell, Goodall,
  King, Niu, Patterson, Pedersen, Skarin, and Schick]{langrock2014b}
Langrock, R., Hopcraft, G., Blackwell, P., Goodall, V., King, R., Niu, M.,
  Patterson, T., Pedersen, M., Skarin, A., and Schick, R. (2014).
\newblock Modelling group dynamic animal movement.
\newblock {\em Methods in Ecology and Evolution\/}, {\bf 5}(2), 190--199.

\bibitem[Lee(2010)Lee]{lee2010}
Lee, A. (2010).
\newblock Circular {D}ata.
\newblock {\em Wiley Interdisciplinary Reviews: Computational Statistics\/},
  {\bf 2}(4), 477--486.

\bibitem[Mardia(1972)Mardia]{mardia72}
Mardia, K.~V. (1972).
\newblock {\em Statistics of Directional Data\/}.
\newblock Academic Press, London.

\bibitem[Mardia and Jupp(1999)Mardia and Jupp]{Merdia1999}
Mardia, K.~V. and Jupp, P.~E. (1999).
\newblock {\em {Directional Statistics}\/}.
\newblock John Wiley and Sons, Chichester.

\bibitem[Mastrantonio {\em et~al.}(2015a)Mastrantonio, Maruotti, and
  Jona~Lasinio]{mastrantonio2015}
Mastrantonio, G., Maruotti, A., and Jona~Lasinio, G. (2015a).
\newblock Bayesian hidden markov modelling using circular-linear general
  projected normal distribution.
\newblock {\em Environmetrics\/}, {\bf 26}, 145--158.

\bibitem[Mastrantonio {\em et~al.}(2015b)Mastrantonio, Jona~Lasinio, and
  Gelfand]{mastrantonio2015b}
Mastrantonio, G., Jona~Lasinio, G., and Gelfand, A.~E. (2015b).
\newblock Spatio-temporal circular models with non-separable covariance
  structure.
\newblock {\em TEST\/}, {\bf To appear}.

\bibitem[McLellan {\em et~al.}(2015)McLellan, Worton, Deasy, and
  Birch]{McLellan2015}
McLellan, C.~R., Worton, B.~J., Deasy, W., and Birch, A. N.~E. (2015).
\newblock Modelling larval movement data from individual bioassays.
\newblock {\em Biometrical Journal\/}, {\bf 57}(3), 485--501.

\bibitem[McMillan {\em et~al.}(2013)McMillan, Hanson, Saunders, and
  Gallun]{McMillan2013}
McMillan, G.~P., Hanson, T.~E., Saunders, G., and Gallun, F.~J. (2013).
\newblock A two-component circular regression model for repeated measures
  auditory localization data.
\newblock {\em Journal of the Royal Statistical Society: Series C (Applied
  Statistics)\/}, {\bf 62}(4), 515--534.

\bibitem[Pewsey(2000)Pewsey]{Pewsey2000}
Pewsey, A. (2000).
\newblock The wrapped skew-normal distribution on the circle.
\newblock {\em Communications in Statistics - Theory and Methods\/}, {\bf
  29}(11), 2459--2472.

\bibitem[Pewsey {\em et~al.}(2013)Pewsey, Neuh\"auser, and Ruxton]{pewsey2013}
Pewsey, A., Neuh\"auser, M., and Ruxton, G.~D. (2013).
\newblock {\em Circular Statistics in R\/}.
\newblock Oxford University Press, Croydon.

\bibitem[Sarma {\em et~al.}(2011)Sarma, Rao, and Girija]{Sarma2011}
Sarma, R., Rao, A. V.~D., and Girija, S.~V. (2011).
\newblock On characteristic functions of the wrapped lognormal and the wrapped
  weibull distributions.
\newblock {\em Journal of Statistical Computation and Simulation\/}, {\bf
  81}(5), 579--589.

\bibitem[Von~Mises(1918)Von~Mises]{vonmises1918}
Von~Mises, R. (1918).
\newblock \"{U}ber die ganzzahligkeit der atomgewicht und verwandte fragen.
\newblock {\em Phys. Z\/}, {\bf 19}, 490--500.

\bibitem[Wang and Gelfand(2014)Wang and Gelfand]{wang2014}
Wang, F. and Gelfand, A.~E. (2014).
\newblock Modeling space and space-time directional data using projected
  {G}aussian processes.
\newblock {\em Journal of the American Statistical Association\/}, {\bf
  109}(508), 1565--1580.

\end{thebibliography}
\end{document}